\documentclass{article}

\usepackage{arxiv}

\usepackage{amsmath,amssymb,amsthm,mathrsfs,amsfonts,dsfont} 

\usepackage[utf8]{inputenc} % allow utf-8 input
\usepackage[T1]{fontenc}    % use 8-bit T1 fonts
\usepackage{hyperref}       % hyperlinks
\usepackage{url}            % simple URL typesetting
\usepackage{booktabs}       % professional-quality tables
\usepackage{amsfonts}       % blackboard math symbols
\usepackage{nicefrac}       % compact symbols for 1/2, etc.
\usepackage{microtype}      % microtypography
\usepackage{lipsum}

\usepackage{algorithm}%,algorithmicx,algpseudocode}
\usepackage{algorithmic}

\usepackage[english]{babel}
\usepackage{subcaption}

\usepackage[square,sort,comma,numbers]{natbib}
 
 \vspace{-10mm}

\usepackage[english]{babel}
\usepackage{subcaption}
\newtheorem{mydef}{Definition}
\usepackage{mathtools}
\newtheorem{theorem}{Theorem}
\theoremstyle{definition}

\newtheorem{corollary}{Corollary}[theorem]

\usepackage{kantlipsum} %<- For dummy text
\newtheorem{asu}{Assumption}
\usepackage{amsfonts}
\usepackage{bbm}
\usepackage{mathptmx}
\newtheorem{remark}{Remark}
\usepackage{titlesec}

\DeclareCaptionLabelSeparator{periodspace}{.\quad}
\title{Differentially Private Collaborative Intrusion Detection Systems For VANETs}

\author{
  Tao Zhang\\
  Department of Electrical and Computer Engineering\\
  New York University\\
  \texttt{tz636@nyu.edu} \\
   \And
 Quanyan Zhu\\
  Department of Electrical and Computer Engineering\\
 New York University\\
  \texttt{qz494@nyu.edu} \\
}

\begin{document}
\maketitle

\begin{abstract}
Vehicular ad hoc network (VANET) is an enabling technology in modern transportation systems for providing safety and valuable information, and yet vulnerable to a number of attacks from passive eavesdropping to active interfering. 
Intrusion detection systems (IDSs) are important devices that can mitigate the threats by detecting malicious behaviors. Furthermore, the collaborations among vehicles in VANETs can improve the detection accuracy by communicating their experiences between nodes. To this end, distributed machine learning is a suitable framework for the design of scalable and implementable collaborative detection algorithms over VANETs. One fundamental barrier to collaborative learning is the privacy concern as nodes exchange data among them. % the data exchange between nodes over the network can impose serious privacy concerns. 
A malicious node can obtain sensitive information of other nodes by inferring from the observed data. In this paper, we propose a privacy-preserving machine-learning based collaborative IDS (PML-CIDS) for VANETs.
The proposed algorithm employs the alternating direction method of multipliers (ADMM) to a class of empirical risk minimization (ERM) problems and trains a classifier to detect the intrusions in the VANETs. We use the differential privacy to capture the privacy notation of the PML-CIDS and propose a method of dual variable perturbation to provide dynamic differential privacy. We analyze theoretical performance and characterize the fundamental tradeoff between the security and privacy of the PML-CIDS. We also conduct numerical experiments using the NSL-KDD dataset to corroborate the results on the detection accuracy, security-privacy tradeoffs, and design.

\end{abstract}

% keywords can be removed

\section{Introduction}

With a growing number of vehicles on road and the rapid development of autonomous vehicles, road safety becomes an increasingly important issue.
Vehicular ad hoc network (VANET) provides a communication system that enables the dissemination of safety-related information, traffic management, navigation, and road services.
%
% that can provide safety-related information, traffic management, navigation, and entertainment services. 
%
However, it is known that VANETs are vulnerable to a number of attacks from passive eavesdropping to active interfering \cite{pathan2016security}.
For example, an attacker can eavesdrop and log the messages of other vehicles, and replay them to access specific resources such as toll services.
An attacker can intrude a specific vehicle, impersonate its identity, and send out false warnings that can disrupt the highway traffic \cite{pathan2016security}.% Other examples include denial of service, hardware tampering, privacy violation, message delay or suppression (e.g., in the case of accidents, safety message can be delayed and suppressed, which can cause immense damage.). Maintaining a high level of security in the VANET is of paramount importance to prevent a number of cyber attacks that lead to major problems and serious dangers to life and property.

Intrusion detection plays an important role in mitigating the threat of VANETs by using \textit{signature}-based and/or \textit{anomaly}-based approaches to detect adversarial behaviors  \cite{zhang2003secure}.  %is an essential part of a security mechanism in the MANETs \cite{zhang2003secure}. 
%Intrusion activities include any use of a system that exceeds authentication limits, and any attempt to access more than it is allowed. 
%
%Intrusion detection can be categorized into \textit{signature}-based, \textit{anomaly}-based, and hybrid ones. Signature-based approaches are designed to detect known attacks according to the signature database. 
%
%IDSs gain a high level of accuracy in detecting known attacks but often fail when unknown attacks occur \cite{buczak2015survey}.% Anomaly-based techniques create models of normal behaviors and detect activities as an anomaly when there are sufficient deviations from the normal behaviors. The advantages of anomaly-based techniques include detecting unknown attacks such as zero-day attacks, and flexibility of profiles of normal behaviors \cite{buczak2015survey}. However, the main disadvantage of anomaly-based detection is the high false positive rate. Hybrid approaches combine the signature-based and the anomaly-based to improve the detection rates of known attacks while decreasing the false positive rate for unknown attacks. 
%
%Many architectures of IDSs have been proposed in the field of MANET (see related work in Section \ref{related}), including stand-alone IDSs, collaborative IDSs (CIDSs), hierarchical IDSs, and mobile-agent-based IDSs \cite{anantvalee2007survey}.
Among many architectures of IDSs, the collaborative IDSs (CIDSs) have been proposed to enable the sharing of detection knowledge about known and unknown attacks and increase detection accuracy \cite{anantvalee2007survey, zhu2012guidex, fung2010bayesian}. %without increasing the complexity of the local IDS as compared to the traditional single host-based ones.
Distributed machine learning algorithms provide an appropriate framework for CIDSs to classify adversarial behaviors using local datasets and share knowledge to increase the detection accuracy.% In recent years, machine learning has started to play a significant role in IDSs. However, it is still challenging to apply machine learning techniques for IDSs \cite{buczak2015survey}. 
%The main challenges include the update frequency of the model and the availability of labeled training dataset. Generally, the IDS models are trained daily \cite{bilge2011exposure}, by the user request \cite{jemili2007framework}, or each time a new intrusion is identified \cite{hansen2007genetic}. Frequent updates of the model can enhance the security level of the VANET, which make the training time of the model important. Therefore, fast learning techniques are necessary for maintaining the high-level security. Besides the difficulty of obtaining labeled training data, the large amount of data required also burdens the capacity of the IDSs. 

In this paper, we consider the network-level intrusion attacks on computer system \cite{raiyn2014survey, hoque2014network} and take advantage of the collaborative nature of the VANETs and design a system architecture of a distributed machine-learning based CIDS over a VANET. The CIDS enables each vehicle to utilize the knowledge of the labeled training data of other vehicles; thus, it boosts the training data size for each vehicle without actually burdening the storage capacity of each vehicle. Also, the laborious task of collecting labeled data can be distributed to all the vehicles in a VANET, thereby reducing the workload of each vehicle. Moreover, the CIDS enables the vehicles to share knowledge of each other without directly exchanging the training data. In addition, the CIDS provides the scalability of the training data processing and improves the quality of decision-making, while reducing the computational cost. The \textit{alternating direction method of multipliers} (ADMM) \cite{boyd2011distributed} is one suitable approach to decentralizing the machine learning problem over a network that allows nodes over the network to share their classification results and yields the optimal classifier achieved under the centralized learning. Despite the distributed feature of the learning algorithm, the data communications between different vehicles can create serious privacy concerns of the training data in each vehicle when %when an adversary aims to infer sensitive information of private datasets of other vehicles.}%, even though the direct data sharing has been circumvented by the ADMM mechanism.} 
%
%for example, the leakage of the data about the position of the vehicle could be utilized by third parties to track the movement of that vehicle, which is a violation of privacy; another example might be that the information included in the training data could be used to profile the driver's behavior and traffic pattern.
%
an adversary can observe the outcome of the learning and extract the sensitive information of the training data of each vehicle. The adversary can either be a vehicle of the VANET which observes its neighboring vehicles or malicious outsider who can observe the outputs of learning.

The lack of privacy protection mechanism often creates barriers for information sharing and disincentives for nodes to achieve collaboration.
Therefore, a privacy-preserving mechanism is important to protect the training data privacy over the network and achieve an effective CIDS. Differential privacy proposed in \cite{dwork2006calibrating} has been a well-defined concept that can provide a strong privacy guarantee by which a change of any single entry of the dataset can only slightly change the distribution of the responses of the dataset.

Therefore, this work proposes a privacy-preserving machine-learning based collaborative IDS (PML-CIDS) for the VANET. We first employ ADMM to construct a distributed empirical risk minimization (ERM) problem over a VANET so that a classifier can be trained in a decentralized fashion to detect whether an activity is normal or attack.
We extend the differential privacy to \textit{dynamic differential privacy} to capture the privacy notation in the distributed machine learning of the CIDS, and propose a privacy-preserving approach, \textit{dual variable perturbation} (DVP). We also investigate the performance of the DVP and characterize the fundamental tradeoff between security and privacy of the PML-CIDS by formulating convex optimization problems and conduct numerical experiments based on the NSL-KDD dataset to demonstrate the optimal design of the privacy mechanism. The main contributions of this paper are summarized as follows:
%
%We capture the distributed and collaborative nature of the VANET and propose a mechanism of an intrusion detection system (IDS) using distributed collaborative-based supervised machine learning in addition to the hybrid detection approach. Machine learning techniques have been utilized in intrusion detection methods for automating the detection \cite{upadhyaya2013hybrid, wankhade2013overview}. 
%
%In this research, we employ the alternating direction method of multipliers (ADMM) to construct a collaborative learning over a VANET, and apply a class of ADMM-based empirical risk minimization (ERM) problems as methods of learning a classifier that can predict whether an activity is normal or an attack. 
%
\begin{itemize}
\item[(i)] We propose a machine-learning-based CIDS architecture to enable the collaborative information exchange and knowledge sharing in VANETs.
\item[(ii)] We use ADMM to capture the distributed nature of a VANET and construct a collaborative learning over a VANET based on a regularized ERM algorithm.
\item[(iii)] We develop the DVP method to perturb the dual variables before minimizing the augmented Lagrange function at each ADMM iteration. The DVP is shown to guarantee dynamic differential privacy in the collaborative learning of the CIDS for a VANET. 
\item[(iv)] We investigate the theoretical performance of the DVP, which is measured by the minimum training data size required to achieve a low error.
\item[(v)] We provide a design principle to find the optimal value of the privacy parameter by solving an optimization problem to manage the tradeoff between security and privacy of a VANET.
\end{itemize}

\begin{figure*}[htpb]
\includegraphics[scale=0.32]{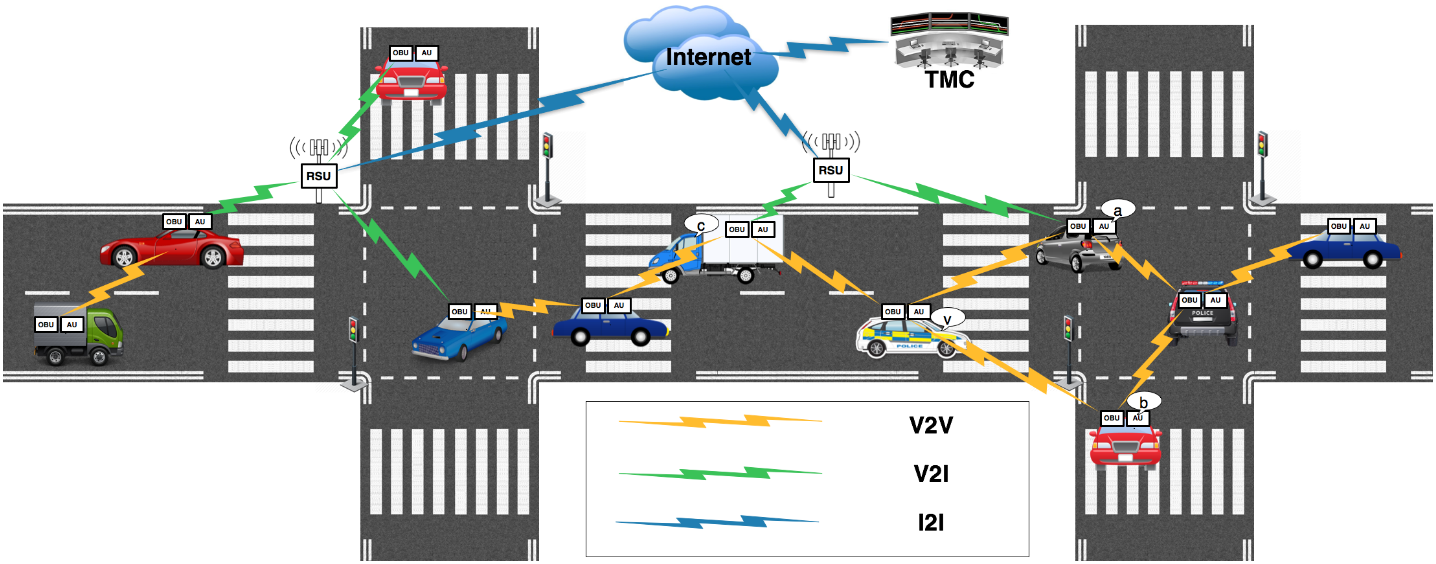}
\centering
\caption{A VANET scenario: TMC: traffic management center; V2V: vehicle-to-vehicle communications; V2I: vehicle-to-infrastructure communications; I2I: infrastructure-to-infrastructure communications. Each vehicle is equipped with an OBU and an AU. } %\vspace{-6mm}
\label{VANET1}
\end{figure*}

\subsection{Related Work}\label{related}
Many works have studied various architectures of intrusion detection systems that are well-suited to MANET \cite{anantvalee2007survey}. Most architectures for MANET can be classified into three categories. The first is the distributed and cooperative IDS, which captures the distributed nature of MANET that has the potential for constructing cooperations over the network. For example, Zhang and Lee in \cite{zhang2003intrusion} have utilized this nature of MANET and constructed a model for a distributed and cooperative IDS. Also, Albers et al. have proposed a collaborative IDS based on local IDS by using mobile agents in \cite{albers2002security}. The local IDS is implemented on each node of the MANET for local node-based security concerns, which can be extended to deal with the global security issues by establishing a collaboration among local IDSs over the MANET. The second category is hierarchical IDS model that extends the distributed and cooperative architectures. In \cite{sterne2005general}, Sterne et al. have designed a dynamic hierarchical IDS using multilevel clustering. The third architecture uses the concept of mobile agents, which can move through the large network. In this type of framework, each mobile agent is assigned to work on a single specific task; then one or multiple mobile agents are distributed into each node in the MANET. 
Previous research includes the work of Kachirski and Guha in \cite{kachirski2003effective} that has proposed distributed IDSs using multiple sensors based on mobile agent technology; and thus, the workload is distributed by separating functional tasks and assigning the tasks to different agents.

Machine learning and data mining for IDSs have also been studied in the literature.
These techniques enable the IDS to continuously learn attacks and their behaviors, enhance the knowledge of the security system, make connections between suspicious events, and predict the occurrence of an attack.
Researchers have studied the unsupervised learning such as the technique of clustering, which is an unsupervised pattern discovery method, in IDSs. There are several approaches for clustering the unlabeled data; for example, Blowers and Williams \cite{blowers2014machine} have applied a density-based spatial clustering of applications with noise clustering algorithm to group normal versus anomalous network packets. Other clustering based work includes hierarchical clustering \cite{horng2011novel} and K-means \cite{muda2011intrusion}. There is also literature on the IDS with supervised learning such as support vector machine \cite{buczak2015survey}. For example, Wagner et al. \cite{wagner2011machine} have applied one-class SVM classifier and used a new window kernel to find the anomaly based on time position of the data. Other methods using supervised learning include decision trees \cite{kruegel2003using, bilge2011exposure, bilge2014exposure}, artificial neural networks \cite{cannady1998artificial, lippmann2000improving}, and sequential data aggregation \cite{fung2016facid, zhu2010distributed}. There also have been works on intrusion-prediction based detection using non-machine-learning techniques \cite{zhu2009dynamic, zhu2011indices}.
For example, Nidhal et al. \cite{mejri2016new} have designed a game-theoretic intrusion detection approach for VANET. The game-based model can predict a possible future denial-of-service attack on the monitored nodes.

In the field of differential privacy, there is a number of works on applying differential privacy to machine learning \cite{kasiviswanathan2011can, bassily2014private, han2014differentially,  zhang2017dynamic, zhang2016dual}. Kasiviswanathan et al. \cite{kasiviswanathan2011can}, for instance, have driven a general method for probabilistically approximately correct learning. 
A body of literature has studied the tradeoff between the privacy and the performance of machine learning while exploring the theory of differential privacy (e.g., \cite{dwork2006calibrating, mcsherry2007mechanism, blum2005practical}). Also, an increasing number of researchers focus on the distributed differential privacy. 
Eigner and Maffei have developed the framework for the automated verification of distributed differential privacy in \cite{eigner2013differential} to enforce the distributed differential privacy in cryptographic protocol implementations. 
Han et al. \cite{Han2016Differentially} have proposed a differentially private algorithm to solve a distributed constrained optimization based on distributed projected gradient descent to protect the privacy of the constraint set.
Hale et al. \cite{Hale2015Differentially} have used a cloud computer to perform differentially private computations so that the broadcasts of the results to each agent over the network do not leak the private state of each agent.
In this paper, we have developed a collaborative IDS using distributed machine learning and resolve the barrier of privacy issues by proposing the concept of dynamic differential privacy to protect the privacy of the training dataset used in the learning.

\subsection{Organization of the Paper}
The rest of the paper is structured as follows. 
Section \ref{Se2} describes the PML-CIDS architecture. Section \ref{Se4} presents the model of the collaborative learning over a VANET for IDS. The ADMM approach is used to decentralize a centralized ERM problem that models the collaborative learning in the VANET. We also describe the privacy concerns associated with the ADMM-based collaborative learning and define the dynamic differential privacy. Section \ref{Se5} proposes the DVP algorithm to provide dynamic differential privacy. Then, we study the performance of the DVP algorithm in Section \ref{Se6}. Section \ref{Se7} shows numerical experiments to corroborate the theoretical results and the optimal design principle to a tradeoff between security and privacy. Finally, Section \ref{Se8} presents the concluding remarks and future research directions.

\section{PML-CIDS Model} \label{Se2}

In this section, we describe the architecture of the proposed PML-CIDS which includes multiple building blocks for VANETs. 
%
%Before introducing our PML-CIDS model, we first describe the general architecture of VANET. 
Illustrated in Figure \ref{VANET1}, a general VANET consists of on-board units (OBU), application units (AU), and roadside units (RSU). The communication between OBUs (vehicle-to-vehicle), or between an OBU and an RSU (vehicle-to-infrastructure) is based on wireless access in-vehicle environment (WAVE) \cite{anantvalee2007survey}. The RSUs can also connect to other infrastructures such as other RSUs and traffic management center, and the communications between them (infrastructure-to-infrastructure) are through other wireless technology. Each vehicle is equipped with an OBU and one or multiple AUs. It also has a set of sensors to collect information and use the OBU to exchange information with other OBUs or RSUs. Details about the three main components of the VANET architecture are presented in the Appendix \ref{Apdx_1} for interested readers.

Each vehicle is equipped with one local PML-CIDS agent as shown in Figure \ref{DDPCIDS} to monitor its local activities including the ones in the AU and the communications via the OBU. Conceptually, the collaborative system consists of three main components, namely, pre-processing engine, a local detection engine, and privacy-preserving collaborative machine learning (P-CML) engine.
The logical flow of a PML-CIDS is illustrated in Algorithm \ref{AlgIDS}.
The pre-processing engine gathers and pre-processes the real-time VANET system data that describe the system activities in a vehicle. The pre-processed system data is then analyzed by the local detection engine using classification techniques. If the user of the vehicle requires the current classifier to be updated, then the P-CML engine is initiated. The local detection engine uses the newly retained classifier to analyze the system data. Otherwise, the current classifier is used in the classification of intrusions. If any intrusion is classified, the alarm is triggered. Each component of the PML-CIDS is elaborated further in Appendix \ref{Apdx_2}. One essential component of the CIDS is the P-CML engine which is composed of the collaborative communication (CC) engine, distributed local learning (DLL), and privacy-preserving (PP) mechanism. The details of these building blocks will be described in detail Section \ref{Se4}.

%

 %%%%%%%%%%%%%%%%%%%%%%%%%%%%%%%%%%
% Algorithm 1
% Distributed ERM
\begin{algorithm}[tb]
   \caption{PML-CIDS}
   \label{AlgIDS}
\begin{algorithmic}
\STATE {\bfseries Input:} Real-time VANET system data: Local audit data flow and activity logs.
\STATE {\bfseries Step 1:} The \textit{\textbf{pre-processing engine}} collects and pre-processes the real-time VANET system data, by numerical transformation, features selection, and data normalization.
\IF{The classifier needs update}
\STATE {\bfseries Step 2:} The \textit{\textbf{P-CML engine}} is initiated and local training dataset is loaded. And updated classifier is obtained.
\STATE {\bfseries Step 3:} The \textit{\textbf{local detection engine}} uses the newly updated classifier to analyze the real-time VANET system data. If any activities are classified as intrusions, the \textit{\textbf{local detection engine}} triggers the alarm.
\ELSE
\STATE {\bfseries Step 2:} The \textit{\textbf{local detection engine}} uses the current classifier to analyze the real-time VANET system data and triggers the alarm when any activities are classified as intrusions.
\ENDIF
\end{algorithmic}
\end{algorithm}

%%%%%%%%%%%% Figure
\begin{figure}[htpb]
\includegraphics[scale=0.20]{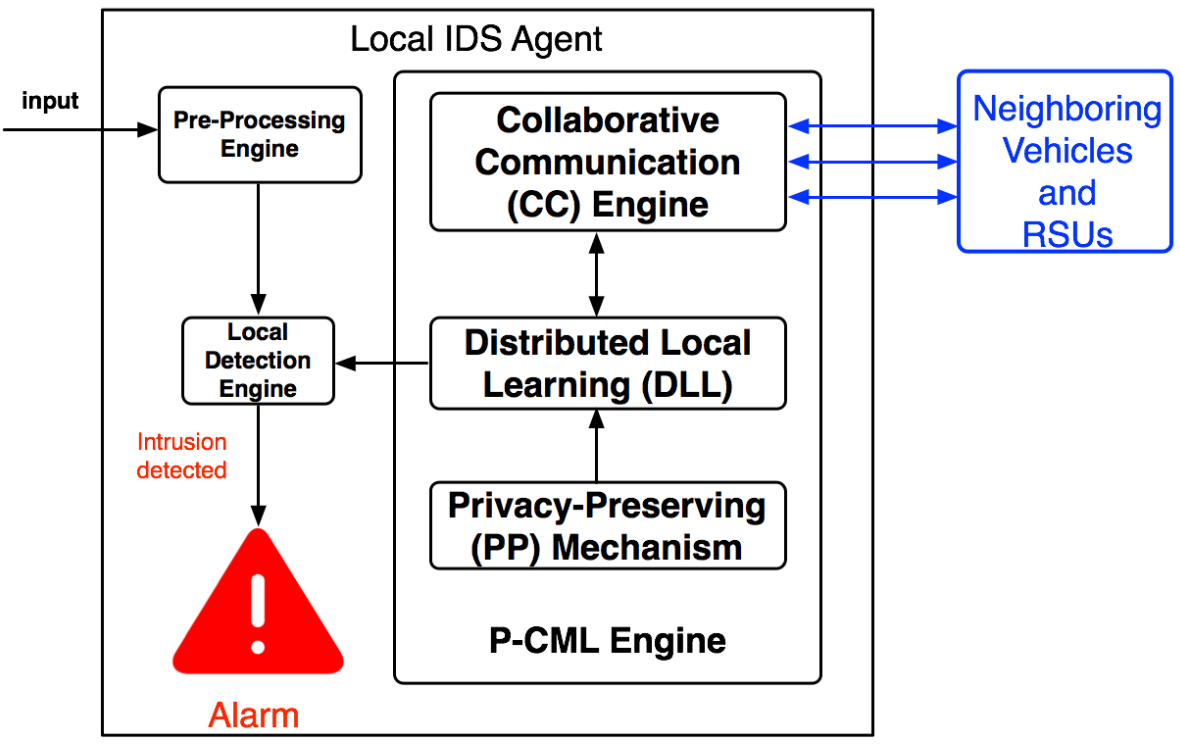}
\centering
\caption{Architecture of PML-CIDS: The Pre-processing engine collects and pre-processes the local audit data flow. The local detection engine then analyzes the pre-processed data using a classifier. If the user of the vehicle requires an update of the classifier, the P-CML engine is initiated. After collaborative learning, the updated classifier will be used in the intrusion detection.} \vspace{-3mm}
\label{DDPCIDS}
\end{figure}

\section{Dynamically Distributed Private Collaborative Learning} \label{Se4}
In this section, we describe the CC engine and the DLL of the P-CML engine in the PML-CIDS. We first model the machine learning by a centralized regularized empirical risk minimization (ERM) problem, which is then decentralized by the ADMM approach. The privacy concerns are then described, and a definition of \textit{dynamic differential privacy} is provided. 
In our model, the vehicles and infrastructures are treated equally except that the infrastructures are static and have more data processing capacity. Therefore, without loss of generality, in the rest of paper, we focus on the vehicle-to-vehicle communications.

% \textbf{\textcolor{red}{\begin{itemize}
% \item $\mathcal{P}\rightarrow \mathcal{V}$
% \item $B_p\rightarrow n_v$
% \item $Z_{c_1} \rightarrow Z_1$
% \item $C^R\rightarrow C_1$
% \item $w_{jp} \rightarrow s_{wv}$
% \item $j \rightarrow w$
% \item $Z_{c_2} \rightarrow Z_2$
% \item $c_1\rightarrow C_2$
% \item $\mu_v\rightarrow\beta_v$
% \item $\hat{C}(f_v)\rightarrow \hat{J}(f_v)$
% \item $\alpha_{acc}\rightarrow \mu$
% \item $\Delta^{non}(t)\rightarrow C^{non}(t)$
% \item $\beta_{non}\rightarrow C_3$
% \item $\beta_{dual}\rightarrow C_4$
% \item $\textbf{J}_{f}\rightarrow \textbf{K}_{f}$
% \item $h_j(\textbf{W})\rightarrow r_j(\textbf{W})$
% \item $f^p(t) \rightarrow f^v$
% \item $\Lambda \rightarrow \Omega$
% \end{itemize}}}

\begin{figure}[htpb]
\includegraphics[scale=0.4]{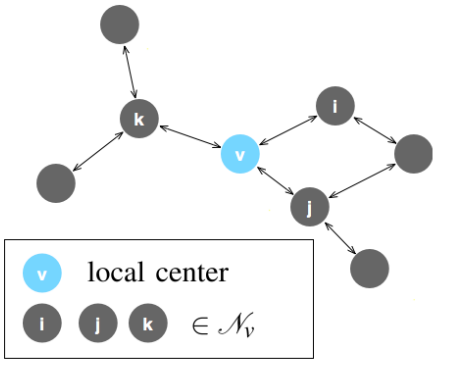}
\centering
\caption{Connected networks: The colored nodes represent the vehicles; $v$ is the local center of the one-hop neighborhood composed of $i$, $j$, and $k$; $v$ only communicates with $i$, $j$, and $k$.} \vspace{-6mm}
\label{VANETabs}
\end{figure}

\subsection{Distributed Learning over a VANET}
Consider a connected VANET, which consists of $P$ vehicles, described by an undirected graph $G(\mathcal{V}, \mathcal{E})$ as shown in Figure \ref{VANETabs} at time $t$, with the set of vehicles $\mathcal{V}=\{1,2,3,..., P\}$, and a set of edges $\mathcal{E}$ denoting the links between connected vehicles.
In general, the graph can change over time as the nodes move. Here we introduce the framework with a fixed topology, and it can be easily extended to dynamic regimes. A particular vehicle $v \in \mathcal{V}$ only exchanges information between its neighboring vehicle $w \in \mathcal{N}_v$, where $\mathcal{N}_v$ is the set of all neighboring vehicles of $v$.
Each vehicle stores a labeled training dataset $D_v = \{(x_{iv},y_{iv})\subset X \times Y : i = 0,1,\cdots,n_{v} \}$ of size $n_{v}$, where $x_{iv} \in X \subseteq \mathds{R}^d$ and $y_{iv} \in Y := \{-1, 1\}$ are the data vector and the corresponding label, respectively. The entire network therefore has a set of data  $\hat{D} = \bigcup_{v \in \mathcal{V}} D_{v}.$ 
The training dataset $D_v$ of $v\in\mathcal{V}$ contains data points describing the VANET system activities such as user and application activities and communication activities through the OBUs; each data point is labeled as an intrusion ($1$) or a normal activity ($-1$). Labeling is important, which is even true for the pure anomaly detection \cite{buczak2015survey}.
The labeled training dataset must have both the normal data and the intrusion data including the novel attack data.

The collaborative learning in our model should be distributed over a VANET without direct data sharing.
The alternating direction method of multipliers (ADMM) is a suitable approach for our model.
In this work, we focus on a class of distributed ADMM-based empirical risk minimization (ERM) as the supervised learning algorithm used in the collaborative learning. 
Each collaborative learning can be modeled as an optimization problem to find a classifier $f: X\rightarrow Y$ using all available data $\hat{D}$ that enable all vehicles in the ad hoc network to classify any input data $x'$ (i.e., data flow collected and pre-processed by pre-processing engine) to a label $y' \in \{-1,1\}$, where $-1$ and $1$ denote normal activities and intrusions, respectively. Before introducing the ADMM-based distributed learning, we first present the centralized optimization in the next subsection.

\subsubsection{Centralized Optimization}
Let $Z_1(f|\hat{D})$ be the centralized objective function of a regularized empirical risk minimization problem (C-ERM). Thus, the C-ERM problem can be defined as:
\begin{equation}\label{CRERM}  
\min_{f}\:\: Z_1(f|\hat{D}) := \frac{C_1}{n_v}\sum_{v=1}^{P}\sum_{i=1}^{n_{v}}\mathcal{\hat{L}}(y_{iv},\:\: f^T x_{iv})+\kappa R(f),
\end{equation}
where $C_1 \leq n_v$ is a regularization parameter and $\kappa>0$ is the parameter that controls the impact of the regularizer.
$\mathcal{\hat{L}}(y_{iv},\:\: f^T x_{iv}): \mathds{R}\times\mathds{R}^d\times\mathds{R}^d\rightarrow \mathds{R}$ is the \textit{loss function} that measures the quality of the trained classifier. In this work, we focus on the specific loss function $ \mathcal{\hat{L}}(y_{iv},\:\: f^T x_{iv})=\mathcal{L}(y_{iv} f^T x_{iv})$. The regularizer function $R(f)$ in (\ref{CRERM}) is used to prevent overfitting.
Suppose that $\hat{D}$ is available to the fusion center vehicle,  a global classifier $f:X\rightarrow Y$ is chosen by optimizing the C-ERM.

\subsubsection{Distributed Optimization}
To solve the problem by ADMM, we first decentralize the C-ERM problem by introducing the decision variables $\{f_{v}\}_{v=1}^{P}$; then, vehicle $v$ determines its own classifier $f_v$. We impose consensus constraints $f_1=f_2=...=f_P$ to guarantee the global consistency of the classifiers. Let $ \{s_{vw}\}$ be the auxiliary  variables to decouple $f_v$ of the vehicle $v$ from its neighbors $w\in \mathcal{N}_v$ in the VANET. 
Then, the consensus-based reformulation of C-ERM becomes 
\begin{equation}\label{equiCRERM}
\begin{split}
\min_{\{f_{v}\}_{v=1}^{P}}\:\: Z_2:=\frac{C_1}{n_v}\sum_{v=1}^{P}\sum_{i=1}^{n_v}\mathcal{L}(y_{iv} f_{v}^T x_{iv})+P\sum_{v=1}^{P}\kappa R(f_v),\\ 
\textrm{s.t. \ \ } f_{v} =s_{vw},\: s_{vw}=f_{v}, v = 1,...,P, w \in \mathcal{N}_v,
\end{split}
\end{equation}
where  $Z_2(\{f_v\}_{v\in\mathcal{V}}|\hat{D})$ is the reformulated objective as a function of $\{f_{v}\}_{v=1}^{P}$. According to Lemma $1$ in \cite{forero2010consensus}, if $\{f_{v}\}_{v=1}^{P}$ presents a feasible solution of (\ref{equiCRERM}) and the network is connected, then problems (\ref{CRERM}) and (\ref{equiCRERM}) are equivalent, i.e., $f = f_{v}$, for all $v = 1,...,P$, where $f$ is a feasible solution of C-ERM. Let $\rho = P\kappa$. Problem (\ref{equiCRERM}) can be solved in a distributed fashion using ADMM with each vehicle $v\in\mathcal{V}$ optimizing the following distributed regularized empirical risk minimization problem (D-ERM):
\begin{equation}\label{equiObjp}
Z_v(f_v|D_v) := \frac{C_1}{n_v}\sum_{i=1}^{n_v}\mathcal{L}(y_{iv}f_v^Tx_{iv}) + \rho R(f_v).
\end{equation}
The augmented Lagrange function associated with the D-ERM is:
\begin{equation}\label{equiDLag}
\begin{aligned}
L^D_v(f_{v},s_{vw}, \lambda^{k}_{vw})=&Z_v+\sum_{i\in \mathcal{N}_v}\big(\lambda_{vi}^{a}\big)^T(f_{v}-s_{vi}) +\sum_{i\in \mathcal{N}_v}\big(\lambda_{vi}^{b}\big)^T(s_{vi}-f_{i}) \\
&+\frac{\eta}{2}\sum_{i\in \mathcal{N}_v}( \parallel f_{v}-s_{vi} \parallel^2 + \parallel s_{vi}-f_{i} \parallel^2).
\end{aligned}
\end{equation}
Therefore, the distributed iterative procedures to solve (\ref{equiObjp}) are:
\begin{equation}\label{equifp1}
f_v(t+1)= \arg\min_{f_v}L_v^D\big(f_v,s_{vw}(t), \lambda^k_{vw}(t)\big),
\end{equation}
\begin{equation}\label{equiAuxi} 
\begin{aligned}
s_{vw}(t+1) = \arg\min_{s_{vw}}L_v^D\big(f_{v}(t+1),s_{vw}, \lambda^k_{vw}(t)\big),
\end{aligned}
\end{equation}
\begin{equation}\label{equiLamb_1_1}
\begin{aligned}
\lambda^a_{vw}(t+1)=\lambda^a_{vw}(t)&+\eta(f_{v}(t+1)-s_{vw}(t+1)),\\
& v\in\mathcal{V},\:\:w\in\mathcal{N}_v,
\end{aligned}
\end{equation}
\begin{equation}\label{equiLambd_1_2}
\begin{aligned}
\lambda^b_{vw}(t+1)=\lambda^b_{vw}(t)&+\eta(s_{vw}(t+1)-f_{v}(t+1)),\\
& v\in\mathcal{V},\:\:w\in\mathcal{N}_v.
\end{aligned}
\end{equation}
Here, $s_{vw}(t+1)$ in (\ref{equiAuxi}) can be found in closed form because the cost in (\ref{equiAuxi}) is linear-quadratic in $s_{vw}(t+1)$\cite{forero2010consensus}. By substituting the closed-form solution, we can eliminate $s_{vw}(t+1)$ from $L_v^D$; this approach makes it possible to simplify the iterative procedures (\ref{equifp1}) to (\ref{equiLambd_1_2}). Indeed, according to Lemma 3 in \cite{forero2010consensus}, we can further simplify the distributed iterative procedures by initializing the dual variables  $\lambda_{vw}^k=\mathbf{0}_{d\times d}$ and combining the two sets of dual variables into one as $\lambda_{v}(t) =  \sum_{w\in \mathcal{N}_v}\lambda^k_{vw} $, $v\in \mathcal{V}$, $w\in \mathcal{N}_v$, $k=a$, $b$. Then, we can combine (\ref{equiLamb_1_1}) and (\ref{equiLambd_1_2}) into one update. We simplify (\ref{equifp1})-(\ref{equiLambd_1_2}) by introducing the following.
Let $L_v^{N}(t)$ be the short-hand notation of $L_v^{N}(\{f_{v}\},\{f_v(t)\}, \{\lambda_{v}(t)\})$ as :
\begin{equation}\label{equiLag_p}
\begin{aligned}
L_v^{N}(t):=&\frac{C_1}{n_v}\sum_{i=1}^{n_v}\mathcal{L}(y_{iv} f_{v}^T x_{iv})+\rho R(f_v) +2\lambda_{v}(t)^{T}f_{v} +\eta\sum_{i\in \mathcal{N}_v}\parallel f_{v}-\frac{1}{2}(f_v(t)+f_i(t)) \parallel^2,
\end{aligned}
\end{equation}
where $\parallel \cdot\parallel$ denotes the $l_2$ norm throughout this paper.

The ADMM iterative procedures (\ref{equifp1})-(\ref{equiLambd_1_2}) are reduced to
\begin{equation}\label{equifp_2}
f_v(t+1)= \arg\min_{f_v}L_v^N(f_v,f_v(t), \lambda_{v}(t)),
\end{equation}
\begin{equation}\label{equiLamb2}
\lambda_{v}(t+1) = \lambda_{v}(t)+ \frac{\eta}{2}\sum_{w\in \mathcal{N}_v}[ f_{v}(t+1)-f_{w}(t+1)].
\end{equation}

% ALGORITHM 1

 % algoithm 1
 
 %%%%%%%%%%%%%%%%%%%%%%%%%%%%%%%%%%
% Algorithm 1
% Distributed ERM
\begin{algorithm}[tb]
   \caption{Distributed ERM over VANET}
 
   \label{Algorithm1}
\begin{algorithmic}
  \STATE {\bfseries Required:} Randomly initialize $f_{v}, \lambda_{v} = \mathbf{0}_{d\times 1}$ for every $v\in\mathcal{V}$ 
   \STATE {\bfseries Input:} $\hat{D}$
   \FOR{$t = 0,1,2,3,...$}
   \FOR{$v=0$ {\bfseries to} $P$}
   \STATE {Compute $f_{v}(t+1)$ via (\ref{equifp_2}).}
   \ENDFOR
   \FOR{$v=0$ {\bfseries to} $P$}
   \STATE {Broadcast $f_{v}(t+1)$ to all neighboring vehicles $w\in\mathcal{N}_v$.}
   \ENDFOR
   \FOR{$p=0$ {\bfseries to} $P$}
   \STATE {Compute $\lambda_{v}(t+1)$ via (\ref{equiLamb2}).}
   \ENDFOR
   \ENDFOR
   \STATE{\bfseries Output:} $f^*=f_v$, for all $v\in \mathcal{V}$.
\end{algorithmic}
\end{algorithm}

Algorithm \ref{Algorithm1} summarizes the (non-private) distributed ERM over a VANET. At iteration $t+1$, vehicle $v$ updates its local $f_v(t)$ through (\ref{equifp_2}). Next, $v$ broadcasts the latest $f_v(t+1)$ to all its neighboring vehicles $w\in\mathcal{N}_v$. When each vehicle has updated $\lambda_v(t+1)$ via (\ref{equiLamb2}), iteration $t+1$ finishes. Throughout the entire algorithm, each vehicle $v\in\mathcal{V}$ only updates its own $f_v(t)$ and $\lambda_v(t)$ and the only information exchanged between neighboring vehicles is $f_v(t)$; thus, direct data sharing is avoided. 
There are several methods to solve (\ref{equifp_2}). For example, projected gradient method, Newton method, and Broyden-Fletcher-Goldfarb-Shanno method \cite{dai2013perfect} that approximates the Newton method, to name a few.
In this distributed algorithm, each vehicle solves a minimization problem per iteration using its local training dataset. The only information in the message transmitted by the OBUs between neighboring vehicles is the value of $f_v(t)$.

ADMM-based distributed machine learning has benefits due to its high scalability. It also provides a certain degree of privacy since vehicles do not share training data directly. However, the privacy issue arises when powerful adversaries can make intelligent inferences at each step of the collaborative learning and extract the privacy information contained in the training dataset based on their observation of the learning output of each vehicle. Simple anonymization or conventional sanitization is not sufficient to address the privacy issue as mentioned in the introduction. In the next subsection, we will discuss the privacy concerns about the training data, and propose differential privacy solutions.

\subsection{Privacy Concerns}\label{Privacy_concerns}

As mentioned in the last subsection, the data stored at each vehicle is not exchanged during the entire ADMM algorithm; however, the potential privacy risk still exists. Suppose that the dataset $D_v$ stored at vehicle $v$ contains sensitive information in data point $(x_s, y_s)$ that is not allowed to be known by anyone else. Consider the worst-case scenario when the adversaries know every data point of training data except $(x_s, y_s)$. There exist risks that the information about the sensitive data can be extracted by observing the output of the non-private ADMM-based distributed learning algorithm when the output is transmitted by OBU.

In this paper, we consider a linear classifier $f_v$. The classifier $f_v$ that minimizes the ERM is a linear combination of data points with the labels, which constitute the entire or a subset of the training dataset, near the decision boundary. Let these data-label pairs constitute a subset of the training dataset $D$, which is denoted as $Sb(D)$.
Let $A_1(\cdot): \mathds{R}^d\rightarrow \mathds{R}$ represent Algorithm \ref{Algorithm1} with the output $f_v = A_1(D_v)$ given the dataset $D_v$.
Let $D'_v$ be any dataset such that $D_v$ and $D'_v$ differ by only one data point. Let 
$(x_{s},y_{s})\subset D_v$ and $(x'_{s},y'_{s})\subset D'_v$ be the only pair of data points that are different, i.e., $(x_{s},y_{s})\neq (x'_{s},y'_{s})$. 
Suppose $f = A_1(D'_v)$.
If $(x_{s},y_{s}) \in Sb(D_v)$ and $(x'_{s},y'_{s}) \in Sb(D'_v)$, then 
$P(A_1(D'_v)=f) = 0$ (thus, $\frac{P(A_1(D_v)=f_v)}{P(A_1(D'_v)=f)}= \infty$).

%%%%%%%%%%%%%%

Before describing the attack model, we first introduce the following notations.
Let $A^r:\mathds{R}^d\rightarrow \mathds{R}$ be the randomized version of Algorithm \ref{Algorithm1}, and let $\{f^*_v\}_{v\in\mathcal{V}}$ be the output of $A_1$. It has been proved (e.g., \cite{forero2010consensus}) that if the number of iterations $t\rightarrow \infty$, $f^*_1= f^*_2=\cdots=f^*_P = f^*$, where $f^*$ is the optimal solution of  C-ERM. Since $A_1$ is deterministic, the output $\{f^*_v\}_{v\in\mathcal{V}}$ is deterministic. In the randomized  algorithm, the vehicle $v$ optimizes its local regularized empirical risk using its own training dataset. 
Let $A^r_{tv}$ be the vehicle-$v$-dependent stochastic sub-algorithm of $A^r$ at iteration $t$, and let $f_v(t)$ be the output of $A^r_{tv}(D_v)$ at iteration $t$ with dataset $D_v$. Therefore, $f_v(t)$ is stochastic at each $t$.

We consider the following attack model. The adversary can access the output at every iteration as well as the final output. This type of adversary aims to obtain the sensitive information contained in the private data point in the training dataset by observing the output $f_v(t)$ or $f^*_v$ for all $v\in \mathcal{V}$ at each iteration $t$, not limited to the first iteration. 
We protect the privacy of distributed learning over a VANET using the definition of differential privacy proposed in \cite{dwork2006calibrating}. Specifically, we require that a change of any single data point in the dataset might only change the distribution of the output of the algorithm slightly. It can be realized by adding randomness to the output of the algorithm. 
Recent advances in the privacy-preserving machine learning techniques are not directly applicable since ADMM algorithms are iterative and dynamic; hence we need to extend the notion of privacy to dynamically differential privacy.
To protect the privacy of training data against the adversary in the collaborative learning of a VANET, we propose the concept of dynamic differential privacy, which enables the D-ERM to be privacy-preserving at every stage of learning.

%%%%%%%%%%

%

\begin{mydef}\label{Def1}(Dynamic $\alpha(t)$-Differential Privacy (DDP))
Consider a network of $P$ nodes $\mathcal{V}=\{1,\:2,\:...,\:P\}$, and each node $v$ has a training dataset $D_v$, and $\hat{D} = \bigcup_{v \in\mathcal{V}} D_{v}$.
Let $ A^r: \mathds{R}^{d} \rightarrow \mathds{R}$ be a randomized version of Algorithm \ref{Algorithm1}.
Let $\alpha(t) = (\alpha_1(t), \alpha_2(t),..., \alpha_P(t))\in \mathds{R}^P_{+}$, where $\alpha_v(t)\in \mathds{R}_+$ is the privacy parameter of node $v$ at iteration $t$.
 Let $A^r_{tv}$ be the node-$v$-dependent sub-algorithm of $A^r$, which corresponds to an ADMM iteration at $t$ that outputs $f_v(t)$.
Let $D'_v$ be any dataset with $H_d(D'_v, D_v)=1$, and $g_v(t)= A^r_{tv}(D'_v)$.
We say that the algorithm $A^r$ is \emph{dynamically $\alpha_v(t)$-differentially private} (DDP) if for any dataset $D'_v$, and for all $v\in\mathcal{V}$ that can be observed by the adversaries, and for all possible sets of the outcomes $S\subseteq \mathds{R}$, the following inequality holds:
\begin{equation}
\Pr[f_v(t)\in S]\leq e^{\alpha_v(t)} \cdot\Pr[g_v(t)\in S],
\end{equation}
for all  $t\in\mathbb{Z}$ during a learning process. The probability is taken with respect to $f_v(t)$, the output of $A^r_{tv}$ at every stage $t$. The algorithm $A^r$ is called \emph{dynamically $\alpha(t)$-differentially private} if the above conditions are satisfied.
\end{mydef}

Definition \ref{Def1} provides a suitable differential privacy concept for the adversary in the collaborative learning of a VANET. For DDP algorithm, the adversaries cannot extract additional information of the private data by observing the $f_v(t)$ at any vehicle $v$ any iteration $t$. As mentioned above, Algorithm \ref{Algorithm1} is not DDP since $\frac{P(A_1(D_v)=f_v)}{P(A_1(D'_v)=f)}\rightarrow \infty$. Please note that the optimization at each iteration in ADMM-based learning is uncoupled from each other different iteration. Also, the optimization at each vehicle is uncoupled from each other.
These properties of ADMM make it possible to treat the privacy of each vehicle each iteration independently. In the definition of DDP, the strength of privacy of vehicle $v$ iteration $t$ totally depends on the value of $\alpha_v(t)$ chosen at $t$, which is independent of the number of $\alpha_w(t')$ for all $w\neq v$ and $t'\neq t$. Therefore, the DDP is also independent of the number of iterations. Since each iteration is private, there are no opportunities for privacy leakage in previous iterations the adversaries can take advantage of to extract more information in later iterations.

%%%%%%%%%%%%
 
\begin{figure*}[htpb]
\includegraphics[scale=0.31]{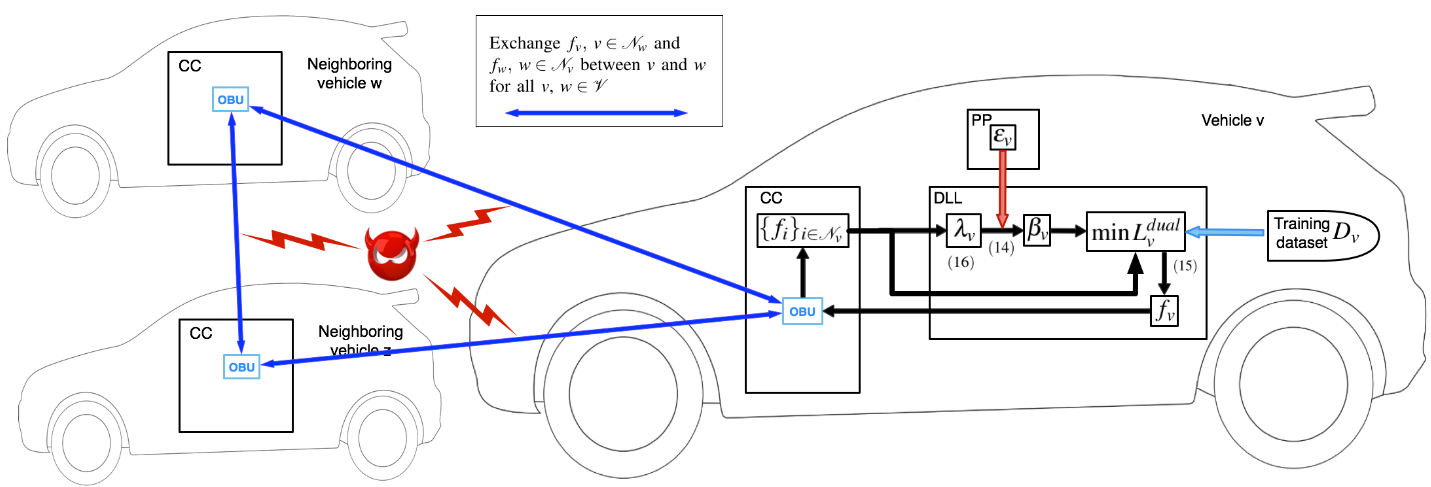}
\centering
\caption{Illustration of DVP during intermediate iterations. The perturbed $\beta_v$ participates in PP mechanism, described by (\ref{equiLamb_dual_pert}). As a result, the output $f_v$ at each iteration is a random variable, and the transmission of $f_v$ is differentially private. The evil red face represents the adversary and the red lighting refer to the possible privacy leakage positions.} \vspace{-6mm}
\label{DVP}
\end{figure*}

\section{Private Collaboration: Dual Variable Perturbation} \label{Se5}

%%%%%%%%%%%%%%%%%%%%%%%%%%%%%%%%%%
% Algorithm 2
% DVP Alg

\begin{algorithm}[tb]
   \caption{Dual Variable Perturbation}
   \label{algDVP}
\begin{algorithmic}[1]
\STATE {\bfseries Required:} Randomly initialize $f_{v}, \lambda_{v} = \mathbf{0}_{d\times 1}$ for every $v\in\mathcal{V}$
   \STATE {\bfseries Input:} $\hat{D}$, $\{[\alpha_v(1),\alpha_v(2),...]\}_{v=1}^{P}$
   \FOR{$t = 0,1,2,3,...$}
   \FOR{$v=1$ {\bfseries to} $P$}
   \STATE {Let $\hat{\alpha}_v = \alpha_v(t) - \ln\Big(1+\frac{C_2}{ \frac{n_v}{C_1}\big(\rho+2\eta N_{v}\big)}\Big)^2$.}
   \IF{$\hat{\alpha}_v >0$}
   \STATE{$ \Phi = 0$.}
   \ELSE
   \STATE{$\Phi= \frac{C_2}{\frac{n_v}{C_1}(e^{\alpha_v(t)/4}-1)}-$
   $\rho-2\eta N_v$ and $\hat{\alpha}_v  = \alpha_v(t)/2$,
   where $N_{v}$ is the number of neighboring vehicles of $v$.}
   \ENDIF
   \STATE{Draw noise $\epsilon_{v}(t)$ according to $\mathcal{K}_v(\epsilon) \sim e^{-\zeta_v(t) \parallel \epsilon \parallel}$ with $\zeta_v(t) = \hat{\alpha}_v $.}
   \STATE{
    {\bfseries \textit{PP:}} Compute $\beta_v(t+1)$  via (\ref{equiLamb_dual_pert}).}
    \STATE{
    {\bfseries \textit{DLL Part 1:}} Compute $f_{v}(t+1)$ via (\ref{equifp_Dual}) with augmented Lagrange function as (\ref{equiLagDual}).}
   \ENDFOR
   \FOR{$v=0$ {\bfseries to} $P$}
   \STATE {
   {\bfseries \textit{CC:}} Broadcast $f_{v}(t+1)$ to all neighboring vehicles $w \in \mathcal{N}_v$.} 
   \ENDFOR
   \FOR{$v=0$ {\bfseries to} $P$}
   \STATE {
   {\bfseries \textit{DLL Part 2:}} Compute $\lambda_{v}(t+1)$ via (\ref{equiLamb_dual}).}
   \ENDFOR
   \ENDFOR
   \STATE{\bfseries Output:} $\{f_v^*\}_{v=1}^{P}$.
\end{algorithmic}
%\AddNote[blue]{1}{5}{Comments for lines 1--5.}
\end{algorithm}
 %%%%%%%%%%%%%%

In the previous section, we have defined a dynamic differential privacy that can capture the notation of data privacy in the collaborative learning over a VANET. In this section, we propose an approach for the privacy-preserving mechanism based on the definition of dynamic differential privacy: \textit{Dual Variable Perturbation} (DVP), and describe the mathematical models of all three components of the P-CML, namely, the PP mechanism, the DLL, and the CC engine.
DVP is proved to be DDP by adding appropriate noise to the deterministic algorithms if the following assumptions are satisfied:

\begin{asu} \label{As1}
The loss function $\mathcal{L}$ is strictly convex and doubly differentiable of $f$ with $| \mathcal{L}'| \leq 1$ and $|\mathcal{L}''|\leq C_2$, where $C_2$ is a constant. Both $\mathcal{L}$ and $\mathcal{L}'$ are continuous.
\end{asu}
\begin{asu} \label{As2}
The regularizer function $R(\cdot)$ is continuous, differentiable, and 1-strongly convex. Both $R(\cdot)$ and $\nabla R(\cdot)$ are continuous.
\end{asu}
\begin{asu} \label{As3}
We assume that $\lVert {x_{iv}} \lVert \leq 1$. Since $y_{iv}\in\{-1,1\}$, $|y_{iv}|=1$.
\end{asu}

Specifically, the DVP provides DDP by perturbing the dual variables $\{\lambda_v(t)\}_{v=1}^{P}$ with a random noise vector $\epsilon_v(t)\in\mathds{R}^d$ with the probability density function
$$
\mathcal{K}_v(\epsilon) \sim e^{-\zeta_v(t) \parallel \epsilon \parallel},
$$
where $\zeta_v(t)$ is a function of $\alpha_v(t)$. In this approach, we add an additional term $\frac{\Phi}{2}\parallel f_v \parallel^2$ to the objective function (\ref{equiObjp}) to make sure that the objective function associated with (\ref{equiLagDual}) is at least $\Phi$-strongly convex. Each iteration starts with perturbing the dual variable $\lambda_v(t)$ updated in the last iteration to a new variable $\beta(t)=\lambda_v(t)+\frac{C_1}{2n_v}\epsilon_v(t)$. Then, the corresponding vehicle-$v$-based augmented Lagrange function $L_v^{N}(t)$ becomes $L_v^{dual}\big(f_{v},f_v(t), \beta_{v}(t+1),\{f_i(t) \}_{i\in\mathcal{N}_v}\big)$. Let $L_v^{dual}(t)$ be a short-hand notation and we have:
\begin{equation}\label{equiLagDual}
\begin{aligned}
L_v^{dual}(t)=& Z_v(f_v|D_v) + \frac{\Phi}{2}\parallel f_v \parallel^2 +2\beta_{v}(t+1)^{T}f_{v} + \eta\sum_{i\in \mathcal{N}_v}\parallel f_{v}-\frac{1}{2}(f_v(t)+f_i(t)) \parallel^2.\\
\end{aligned}
\end{equation}
\noindent Thus, the randomness caused by adding noise $\epsilon_v(t)$ randomizes the minimizer of $L_v^{dual}(t)$. 
Let $C^L_K$ denote the $K$-th (randomized) collaboration. Suppose each $C^L_K$ includes $T_K$ iterations ($T_K$ can be varying in $K$). In our model, the $t$-th iteration of $C^L_K$ is varying in $K$ for all $t\in\{1,\cdots, T_K\}$. This is because that the mechanism at $t$ involves $f_v(t-1)$ and $\beta_v(t-1)$, whose values that were updated in the last iteration $t-1$ are varying due to the random noise at each $t-1$ and at each $K$.

Now we can model three components of the P-CML as:
\begin{itemize}
\item \textbf{PP mechanism: }dual variable $\lambda_v(t)$ is perturbed using Laplace noise $\epsilon_v(t+1)$,
\begin{equation}\label{equiLamb_dual_pert}
\beta_v(t+1) = \lambda_{v}(t)+ \frac{C_1}{2n_v}\epsilon_v(t+1).
\end{equation}
\item \textbf{DLL part 1: }$f_v(t+1)$ is updated by minimizing $L_v^{dual}(t)$,% 18
\begin{equation}\label{equifp_Dual}
f_v(t+1) = \arg\min_{f_v}L_v^{dual}(t).
\end{equation}
\item \textbf{CC engine: }Broadcast $f_v(t+1)$ to all neighboring vehicles $w\in\mathcal{N}_v$.
\item \textbf{DLL part 2: } $\lambda_v (t + 1)$ is updated using all the $\{f_w(t+1)\}_{w\in\mathcal{N}_v}$ from the neighboring vehicles and $f_v(t+1)$,
\begin{equation}\label{equiLamb_dual}
\lambda_v(t+1) = \lambda_v(t)+ \frac{\eta}{2}\sum_{w\in \mathcal{N}_v}[ f_v(t+1)-f_w(t+1)].
\end{equation}
\end{itemize}

The iterations (\ref{equiLamb_dual_pert})-(\ref{equiLamb_dual}) are summarized in Algorithm \ref{algDVP}. Specifically, we have introduced an additional privacy parameter $\hat{\alpha}_v= \alpha_v(t) - 2\ln\Big(1+\frac{C_2}{ \frac{n_v}{C_1}\big(\rho+2\eta N_v\big)}\Big),$ where $N_{v}$ is the number of neighboring vehicles of $v$. $\hat{\alpha}_v$ is required in the proof of Theorem \ref{theorem1}. Two cases of $\hat{\alpha}_v$ are necessary to find the upper bound of the ratio of the Jacobian matrices of the transformation from $f_v(t)$ to $\epsilon_v(t)$ given different datasets (see details in Appendix A in \cite{DBLP:journals/corr/ZhangZ16}).   
Figure \ref{DVP} illustrates each iteration of Algorithm \ref{algDVP}. After the P-CML engine has established a collaboration with neighboring vehicles over a VANET, the ADMM iterations begin. Initially, the DLL at each vehicle generates an initial $\lambda_v(0) = 0$ and a random initial $f_v(0)$. $f_v(0)$ is shared with neighboring vehicles via the CC engine. Every vehicle $v\in\mathcal{V}$ determines its own value of $\rho$ and updates its local parameters $\beta_v(t)$, $f_v(t)$ and $\lambda_v(t)$ at time $t$. At iteration $t+1$, the PP mechanism of $v$ perturbs the $\lambda_v(t-1)$ by a Laplace noise $\epsilon_v(t)$ to generate $\beta_v(t)$ as shown in (\ref{equiLamb_dual_pert}). 
Then, the DLL uses $\beta_v(t)$, $\{f_w(t-1)\}_{w\in\mathcal{N}_v}$, and the local labeled training dataset $D_v$ to update the $f_v(t)$ as shown in (\ref{equifp_Dual}). The CC engine transmits $f_v(t)$ to all the neighboring vehicles, and at the same time, it receives $\{f_w(t)\}_{w\in\mathcal{N}_v}$. Each iteration resumes after the DLL has updated the $\lambda_v(t)$ using $\{f_w(t)\}_{w\in\mathcal{N}_v}$, $f_v(t)$, and $\lambda_v(t-1)$ according to (\ref{equiLamb_dual}). After the P-CML terminates, it transmits the final updated classifier $f^*_v$ to the local detection engine for intrusion detection. The privacy guarantee of DVP is summarized in Theorem \ref{theorem1}.

\begin{theorem} \label{theorem1}
Under Assumption \ref{As1}, \ref{As2} and \ref{As3}, Algorithm \ref{algDVP} solving D-ERM is dynamically $\alpha(t)$-differentially private with $\alpha(t) = (\alpha_1(t),\alpha_2(t),...,\alpha_P(t))$, where $\alpha_v(t)$ is chosen by each vehicle $v\in\mathcal{V}$ at time $t$. Let $Q(f_v(t)|D_v)$ and $Q(f_v(t)|D'_v)$ be the probability density functions of $f_v(t)$ given dataset $D_v$ and $D'_v$, respectively, with $H_d(D_v, D'_v)=1$. The ratio of conditional probabilities of $f_v(t)$ is bounded as follows:
% 20
\begin{equation}\label{equi_Theorem1_bound}
\frac{Q(f_v(t)|D_v)}{Q(f_v(t)|D'_v)}\leq e^{\alpha_v(t)}.
\end{equation}

\end{theorem}

\begin{proof}
See Appendix A in \cite{zhang2017dynamic}.
\end{proof}

\begin{remark}%[\textbf{1}]
In practice, the VANET topology frequently changes due to the mobility of the vehicles. The change of topology can be caused by the changes of the position, the speed, and the number of vehicles. Therefore, it is possible that the VANET topology changes during the collaborative learning.
Section \ref{Privacy_concerns} has explained the independence of dynamical differential privacy. Since the dynamic differential privacy of the training dataset at $v$ is independent of other iterations and the number of iterations, we can conclude that the DVP algorithm is independent of the speed of the vehicles. Also, the privacy of $v$ is independent of the activities at other vehicles; thus the DVP at $v$ is also independent of the mobility of vehicles in the VANET. Therefore, the dynamic differential privacy and the Algorithm \ref{algDVP} work in the topology-varying VANET. Let $n_v(t)$ represent the time-varying number of vehicles in the topology-varying VANET. Let $\mathcal{N}_v(t)$ and $N_v(t)$ denote the time-varying set of neighboring vehicles and the number of neighboring vehicles, respectively. By substituting the time-varying $n_v(t)$, $\mathcal{N}_v(t)$ and $N_v(t)$ into equations (\ref{equiLamb_dual_pert})-(\ref{equiLamb_dual}), the algorithm is also dynamically differentially private in the topology-varying VANET.
\end{remark}

%%%%%%

\begin{remark}
In this model, the learning is a continuous progress and we do not specify a time window of learning for each vehicle $v$. Specifically, each $v$ decides when to start a new collaborative learning to update the previously updated classifier $f^1_{v}$, or to stop a collaborative learning in progress and keep the newly updated classifier $f^2_{v}$ as the latest intrusion classifier.  Continuous learning is important since the training data keeps being updated. The machine learning algorithm can benefit from the frequent changes of dataset to continuously learn different kinds of attacks and their behaviors and enhancing the knowledge of the security system.
\end{remark}

%%%%%%

\section{Performance Analysis} \label{Se6}
In this section, we discuss the performance of Algorithm \ref{algDVP}. The training data stored at each vehicle is labeled as $1$ (attack) or $-1$ (normal). We consider two types of errors: false positive error and false negative error. The false positive (or false negative) refers to labeling a data point $x''$ as $1$ (or $-1$) when actually $x''$ is $-1$ (or $1$).
We establish performance bounds for the $l_2$ norm regularization functions such that we can train a classifier with low false positive and low false negative.  
The performance analysis is based on the following assumptions:
\begin{asu} \label{As4}
The data-label pair $\{(x_{vi}, y_{vi})\}_{i=1}^{n_v}$ are drawn i.i.d. from a fixed but unknown probability distribution $\mathds{P}^{xy}(x_{vi}, y_{vi})$ at each node $v\in\mathcal{V}$. Also, there is fixed but unknown conditional probability distribution $\mathds{P}^{x|y}(x_{vi}| y_{vi}=q)$ for data points $\{x_{vi}\}_{i=1}^{n_v}$ given $y_{vi}=q$, where $q = -1$ or $1$. 
\end{asu} 
\begin{asu} \label{As5}
$\epsilon_v(t)$ is drawn from $
\mathcal{K}_v(\epsilon) \sim e^{-\zeta_v(t) \parallel \epsilon \parallel},
$ with the same $\alpha_v(t)=\alpha(t)$ (thus the same $\zeta_v(t)$) for all $v\in\mathcal{V}$ at time $t\in\mathbb{Z}$. 
\end{asu} 

According to Assumption \ref{As4}, we define the conditional expected loss function of the classifier $f_v$ of vehicle $v$, given $y$ as:
$
\hat{J}(f_v|y) :=C_1\mathds{E}_{x\sim \mathds{P}^{x|y}}(\mathcal{L}(yf^Tx)|y);
$  
thus, the corresponding conditional expected objective function $\hat{Z}_v$ is
$
\hat{Z}_v(f_v|y) := \hat{J}(f_v|y) + \rho R(f_v).
$

The performance of non-private centralized ERM classification optimization has been already studied in the literature (e.g., \cite{shalev2008svm, chaudhuri2011differentially}). For example, Shalev et al. in \cite{shalev2008svm} introduces a reference classifier $f^0$, and shows that there is a lower bound of the training data size such that the actual (unconditional) expected loss of the $l_2$ regularized support vector machine (SVM) classifier $f_{SVM}$ satisfies
$\hat{J}(f_{SVM})\leq \hat{J}^0 + \mu,$ where $\mu$ is the generalization error and $\hat{J}^0=\hat{J}(f^0)$. The similar argument can be used in this work to study the accuracy of Algorithm \ref{Algorithm1} in terms of conditional expected loss. Let $\hat{J^0}_{x|y=q} = \hat{J}(f^0|y=q)$.
We quantify the performance of Algorithm \ref{Algorithm1} with the final output $f^*$ by the minimum number of data points required to obtain
$
\hat{J}(f^*|y=q) \leq \hat{J^0}_{x|y=q}+\mu_q.
$

However, instead of focusing on only the final output $f^*=\arg\min_{f_v}Z_v(f_v|D_v,y=q)$, for all $v\in\mathcal{V}$, we also care about the performance of the output of each iteration. Let $f_v^{non}(t+1)=\arg\min_{f_v}L_v^N(t)$ be the output of iteration $t$ of the (non-private) Algorithm \ref{Algorithm1} at vehicle $v$. 
Literature has proved that the sequence $\{f_v^{non}(t)\}$ is bounded and converges to $f^*$ as time $t\rightarrow \infty$ (e.g., \cite{forero2010consensus}). Thus, there exists a constant $C^{non}_{q}(t)$ at time $t$ such that
$
\hat{J}(f_v^{non}(t)|y=q) - \hat{J}(f^*|y=q) \leq C^{non}_{q}(t),
$
and substituting it to 
$
\hat{J}(f^*) \leq \hat{J}^0+\mu_q
$
yields:
\begin{equation}\label{section4_2}
\hat{J}(f_v^{non}(t)|y) \leq \hat{J}_{x|y=q}^0+C^{non}_{q}(t) +\mu_q.
\end{equation}
As shown later in this section, the training data size depends on the $\parallel f^0 \parallel$. Usually, the reference classifier is selected with an upper bound on $\parallel f^0 \parallel$. Theorem \ref{theorem2} summarizes the performance analysis of Algorithm \ref{Algorithm1} based on (\ref{section4_2}).
%
%
% Theorem 2: Algorithm 1
\begin{theorem} \label{theorem2}
Let $R(f_v(t))=\frac{1}{2}\parallel f_v(t)\parallel^2$, and let $f^0_1$ and $f^0_{-1}$ such that $\hat{J}(f_1^0|1)=\hat{J}^0_{x|1}$ and $\hat{J}(f_{-1}^0|-1)=\hat{J}^0_{x|-1}$, respectively, for all $v\in\mathcal{V}$ at time $t$, and $\delta_q >0$ is a positive real number for $q=1$ and $-1$. Let $D_v=\Big\{(x_{iv},y_{iv})\subset \mathds{R}^d \times \{-1,1\}\Big\}$ be the dataset of vehicle $v\in\mathcal{V}$. Let $D_v^{(1)}$ and $D_v^{(-1)}$ be the dataset containing all the data points $x_{vi}$ labeled as $1$ and $-1$, respectively; and let $n^{(1)}_v$ and $n^{(-1)}_v$ be the size of $D_v^{(1)}$ and $D_v^{(-1)}$, respectively; thus $D_v= D_v^{(1)}\cup D_v^{(-1)}$, and $n_v = n_v^{(1)}+n_v^{(-1)}$.
Let $f_v^{non}(t+1) = \arg\min_{f_v} L_v^N(f_v,t|D_v)$ be the output of Algorithm \ref{Algorithm1}. 
If Assumption \ref{As1} and \ref{As4} are satisfied, then there exist two constants $C^{(1)}_3$ and $C^{(-1)}_3$ such that if $n^{(1)}_v$ and $n^{(-1)}_v$  satisfy
$$
n^{(1)}_v>C^{(1)}_3\Bigg( \frac{C_1\parallel f_1^0 \parallel^2 \ln(\frac{1}{\delta_1}) }{\mu_1^2}  \Bigg),
$$
and
$$
n^{-1}_v>C^1_3\Bigg( \frac{C_1\parallel f_{-1}^0 \parallel^2 \ln(\frac{1}{\delta_{-1}}) }{\mu_{-1}^2}  \Bigg),
$$
then $f_v^{non}(t+1)$ satisfies
$$
\mathds{P}\big( \hat{J}(f_v^{non}(t+1)|y=q)\leq \hat{J}^0_{x|y=q}+\mu_q+C^{non}_{q}(t) \big)\geq 1-\delta_q,
$$
for all $t\in\mathbb{Z}_+$. Therefore, both false positive and false negative errors are bounded with probability at least $1-\delta_q$.

\end{theorem}

\begin{proof}
See Appendix D in \cite{zhang2017dynamic}.
\end{proof}

Usually, $\mu_q\leq 1$ is required for most machine learning algorithms. In the case of SVM, if the constraints are  $y_i f^Tx_i\leq C_{SVM}$, for $i = 1,\:,...,\:n_{SVM}$, the classification margin is $ \frac{C_{SVM}}{\parallel f^0 \parallel}$. Therefore, if we want to maximize the margin $ \frac{C_{SVM}}{\parallel f^0 \parallel}$, a large value of $\parallel f^0 \parallel$ is required. A larger value of $\parallel f^0 \parallel$ is usually chosen for non-separable or small-margin problems. 
In the following subsection, we use the similar analysis for the performance of Algorithm \ref{algDVP}.

\subsection{Performance of DVP}
Similarly, Algorithm \ref{algDVP} solves one optimization problem minimizing $L_v^{dual}(f_v,t|D_v)$ at each iteration $t$ vehicle $v$. Suppose at iteration $\tau$, we generate a noise term $\epsilon_v(\tau) = \epsilon$. If we fix the noise term $\epsilon_v(t') = \epsilon$ for all $t'>\tau$, then the algorithm becomes static starting from $\tau$. Let Alg-2 denotes this corresponding static algorithm associated with Algorithm \ref{algDVP}. Therefore, solving Alg-2 is equivalent to solving the optimization problem with the object function $Z_v^{dual}(f_v, \tau|D_v, \epsilon)$, defined as follows:
$$
\begin{aligned}
Z_v^{dual}(f_v, \tau|D_v, \epsilon) :=Z_v(f_v|D_v)+\frac{C_1}{n_v}\epsilon f_v.
\end{aligned}
$$ 
Let $Z_v^{dual}(\tau)$ be the short-hand notation of $Z_v^{dual}(f_v, \tau|D_v, \epsilon)$. Note that the index $\tau$ indicates that this objective is based on the noise $\epsilon_v(\tau)=\epsilon$ generated at iteration $\tau$ of Algorithm \ref{algDVP}.
Let $f'_v(t)$ and $\lambda'_v(t)$ be the primal and dual updates, respectively, of the ADMM-based algorithm minimizing $Z_v^{dual}(\tau)$ at iteration $t$. Then, Alg-2 can be interpreted as minimizing $Z_v^{dual}(\tau)$ with $f'_v(0)=f_v(\tau)$ and $\lambda'_v(0)=\lambda_v(\tau)$ as initial conditions for all $v\in\mathcal{V}$. 
Since $Z_v^{dual}(\tau)$ is real and convex, similar to Algorithm \ref{Algorithm1}, the sequence $\{f'_v(t)\}$ is bounded and $f'_v(t)$ converges to $f^{'*}_v(\tau) =  \arg\min_{f'_{v}}Z_v^{dual}(\tau)$, which is a limit point of $f'_v(t)$. 
Therefore, there exists a constant $C^{dual}_{v,y=q}(t)$ given noise term $\epsilon$ fixed from iteration $\tau$ of Algorithm \ref{algDVP} such that 
$$
\hat{J}(f_v(\tau)|y=q) - \hat{J}({f'}^*(\tau)|y=q)\leq C^{dual}_{v,y=q}(t),
$$
for $q=1$ or $-1$.
The way we analyze the performance in Theorem \ref{theorem2} can also be used in the case of DVP. Specifically, the performance is measured by the training data sizes, $n^{(1)}_v$ and $n^{(-1)}_v$ for data points $x_{vi}$ labeled by $y_{vi}=1$ and $y_{vi}=-1$, respectively, for all $v\in \mathcal{V}$ required to obtain
$$
\hat{J}(f_v(\tau)|y=q) \leq \hat{J}^0_{x|y=q}(\tau)+\mu_q+C^{dual}_{v,y=q}(\tau),
$$
for $q=1$ and $-1$.
We say that each $f_v(\tau)$ is accurate with low false positive ($1$) or false negative ($-1$) error if it satisfies the above inequality.
The analysis of the performance for Algorithm \ref{algDVP}, the DVP, is summarized in Theorem \ref{theorem3} and Corollary \ref{corollary41}.

\begin{theorem}\label{theorem3}
Let $R(f)=\frac{1}{2}\parallel f\parallel^2$, and let $f^0_{v,y=1}(\tau)$ and $f^0_{v,y=-1}(\tau)$ such that 
$\hat{J}(f^0_{y=1}(\tau)|1)=\hat{J}^0_{x|y=1}(\tau)$ and $\hat{J}(f^0_{y=-1}(\tau)|-1)=\hat{J}^0_{x|y=-1}(\tau)$ for all $v\in\mathcal{V}$. Let $\delta_1$ and $\delta_{-1}$ be positive numbers. Let $D_v= D_v^{(1)}\cup D_v^{(-1)}=\Big\{(x_{iv},y_{iv})\subset \mathds{R}^d \times \{-1,1\}\Big\}$ be the labeled dataset of vehicle $v\in\mathcal{V}$, where $D_v^{(1)}$ and $D_v^{(-1)}$ are the datasets of size $n^{(1)}_v$ and $n^{(-1)}_v$, respectively, containing all the data points $x_{vi}$ labeled as $1$ and $-1$, respectively.
If Assumption \ref{As1}, \ref{As4} and \ref{As5} are satisfied, then there exist two constants $C^{(1)}_4$ and $C^{(-1)}_4$ such that if the number of data points $n^{(q)}_v$ satisfies
$$
\begin{aligned}
n^{(q)}_v > C^{(q)}_4\max \Bigg( & \max_{\tau} \Big( \frac{\parallel f^0_{v,y=q}(\tau) \parallel d \ln(\frac{d}{\delta_q})}{\mu_q\alpha_v(\tau)}\Big),\max_t \Big(\frac{C_1C_2\parallel f_{v,y=q}^0(\tau) \parallel^2}{\mu_q\alpha_v(\tau)}\Big), \max_t \Big(\frac{C_1\parallel f^0_{v,y=q}(\tau) \parallel^2 \ln(\frac{1}{\delta_q}) }{\mu_q^2}\Big) \Bigg),
\end{aligned}
$$
then $f^*_v(\tau)$ satisfies
$$
\mathds{P}\big( \hat{J}(f^*_v(\tau)|y=q)\leq \hat{J}^0_{x|y=q}(\tau)+\mu_q  \big)\geq 1-2\delta_q,
$$
for all $q =1$ and $-1$.
\end{theorem}

\begin{proof}
See Appendix E in \cite{zhang2017dynamic}.
\end{proof}

\begin{corollary} \label{corollary41}
\textit{Let $f_v(\tau) = \arg\min_{f_v} L_v^{dual}(f_v,\tau-1|D_v)$ be the updated classifier of Algorithm \ref{algDVP} and let $f^0_{v, y=q}(\tau)$ be a reference classifier such that $\hat{J}(f^0_v(\tau)|y=q)=\hat{J}_{x|y=q}^0(\tau)$ for $q=1$ and $-1$. If all the conditions of Theorem \ref{theorem3} are satisfied, then $f_v(\tau)$ satisfies \begin{equation}\label{Coro_DVP_Conv}
\mathds{P}\big(\hat{J}(f_v(\tau)|y=q)\leq \hat{J}^0_{x|y=q}(\tau)+\mu_q +C_{v,y=q}^{dual}(\tau) \big)\geq 1-2\delta_q,
\end{equation}}
for all $q =1$ and $-1$.
\end{corollary}
\begin{proof}
The inequality $\hat{J}(f_v(\tau)|y=q) - \hat{J}(f^*_v(\tau)|q)\leq C^{dual}_{v,y=q}(\tau)$
holds for $f_v(\tau)$ and $f^*_v(\tau)$, for $q=1$ and $-1$, and from Theorem 3, 
$ \mathds{P}\big( \hat{J}(f^*_v(\tau)|y=q)\leq \hat{J}_{x|y=q}^0(\tau)+\mu_q  \big)\geq 1-2\delta_q, $
for $q=1$ and $-1$.
Therefore, we can have (\ref{Coro_DVP_Conv}).
\end{proof}

\section{Numerical Experiments}\label{Se7}

% This Figure 5
\begin{figure*}%
\centering
\begin{subfigure}{.5\columnwidth}
\includegraphics[width=\columnwidth]{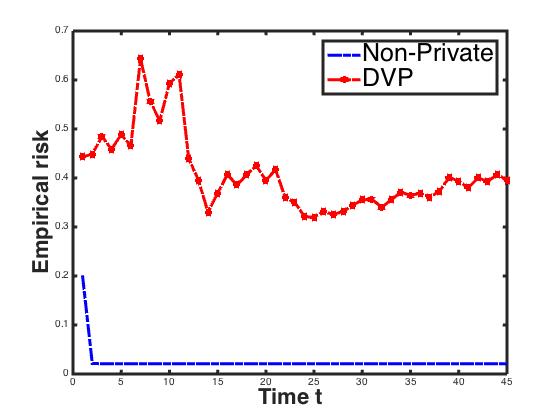}%
\caption{DVP: $\alpha_v = 0.01$}%
\label{conv_1}%
\end{subfigure}\hfill% 
\begin{subfigure}{.5\columnwidth}
\includegraphics[width=\columnwidth]{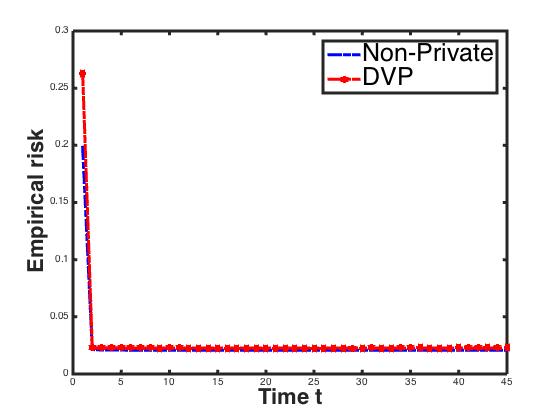}%
\caption{DVP: $\alpha_v = 0.5$}%
\label{conv_2}%
\end{subfigure}\hfill%
\begin{subfigure}{.5\columnwidth}
\includegraphics[width=\columnwidth]{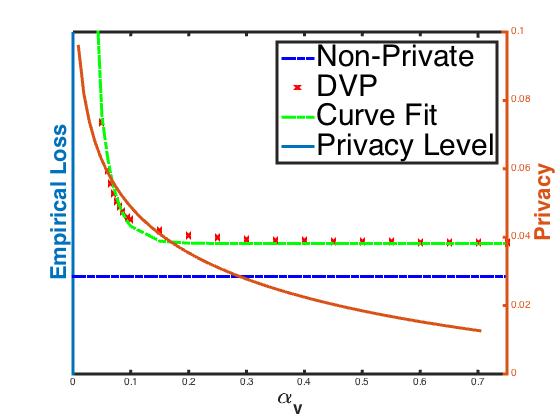}%
\caption{Tradeoff: empirical risk}%
\label{tradeoffRisk}%
\end{subfigure}\hfill%
\begin{subfigure}{.5\columnwidth}
\includegraphics[width=\columnwidth]{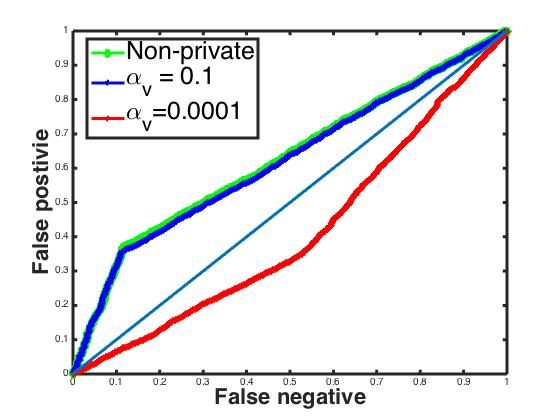}%
\caption{ROC}%
\label{MisE}%
\end{subfigure}\hfill%
\caption{Figure \ref{conv_1}-\ref{conv_2}: convergence  with different (fixed) values of $\alpha_v(t)$; DVP with $\rho = 10^{-2.5}$, $C_1=650$; non-private (Algorithm \ref{Algorithm1}) with  $\rho = 10^{-9}$, $C_1=1250$. Figure \ref{tradeoffRisk}: security-privacy tradeoff measured by the false positive and false negative error rates with $C_{v1} = 20$, $C_{v2}=6$, $C_{v3}=5$, $C_{v4}=1$.
Figure \ref{MisE}:Receiver operating characteristic (ROC) curve for non-private, DVP with different values of $\alpha_v$; DVP with $\rho = 10^{-2.5}$, $C_1=650$; non-private (Algorithm \ref{Algorithm1}) with  $\rho = 10^{-9}$, $C_1=1250$.}
\label{figTradeoff}\vspace{-3mm}
\end{figure*}

In this section, we test the learning performance of Algorithm \ref{algDVP} and explore the tradeoff between security and privacy. We simulate the user and system activities, the communication activities in the AUs and OBUs of the VANET based on the \textit{NSL-KDD} data,
which is the refined version of its predecessor of \textit{KDD'99} and solves some of the inherent problems of the KDD'99 \cite{tavallaee2009detailed}. The NSL-KDD dataset contains essential records of the complete KDD dataset. Each record contains $41$ attributes indicating different features of flow with a label assigned either as an attack or normal. 
Due to the lack of public datasets for network-based IDSs, the NSL-KDD is currently the best available dataset for benchmarking of different intrusion detection methods \cite{tavallaee2009detailed, buczak2015survey}.
%Due to the lack of public datasets for network-based IDSs, the NSL-KDD is current the best available dataset and still can be regarded as an effective benchmark dataset to be used in the researches of different intrusion detection methods, though it is a relatively old dataset that does not include cases of all new attacks that have occurred since 1999 \cite{tavallaee2009detailed, buczak2015survey}.

%
In the experiments, the task is to classifier whether a network activity is an attack ($1$) or normal ($-1$) using \textit{logistic regression}. There are four types of attacks presented in NSL-KDD, namely, denial of service, probing, unauthorized access to local system administrator privileges, and unauthorized access from a remote machine \cite{dhanabal2015study}. In this experiment, we only classify whether an activity is an attack or normal without identifying the specific type of the attack.

We also propose an approach to select an optimal value of $\alpha_v(t)$ that can manage the tradeoff between security and privacy by introducing a utility function of privacy. In the experiments, we fix the value of $\alpha_v(t)$ for each entire running of Algorithm \ref{algDVP}; thus, the noise of each vehicle $v\in\mathcal{V}$ generated at each running of DVP is i.i.d.

To process the NSL-KDD dataset into a form suitable for the classification learning and satisfying the Assumption \ref{As3}, we process the NSL-KDD dataset according to the procedures suggested in \cite{mohamad2015hybrid}. The main processes include the transformation of symbolic attributes to numeric values, feature selection that eliminates irrelevant, noisy or redundant features, data normalization that helps speed up the learning.

When the P-CML engine in vehicle $v$ is initiated, collaboration is established over the VANET. As shown in Figure \ref{VANET1}, the vehicle $v\in\mathcal{V}$ only communicates with vehicles in its one-hop neighborhood composed of three vehicles, $a$, $b$, and $c\in \mathcal{N}_v$, which also communicate with neighboring vehicles directly or through the RSU (e.g., $a$ and $c$). Each vehicle in the collaboration updates its own primal and dual parameters simultaneously.

\subsection{Logistic Regression}
In the experiments, we test the DVP-based Algorithm \ref{algDVP} using logistic regression. Let $\mathcal{L}_{lr}$ be the loss function of logistic regression, which has the form
$$
\begin{aligned}
\mathcal{L}_{lr}(y_{iv}f^T x_{iv})= \log(1+\exp(-y_{iv} f_v^T x_{iv})).
\end{aligned}
$$ 
Clearly, the first and the second order derivatives of $\mathcal{L}_{lr}$ can be bounded as $|\mathcal{L}'_{lr}|\leq 1$ and $|\mathcal{L}''_{lr}|\leq \frac{1}{4}$, respectively.
Therefore, the logistic regression satisfies the conditions in Assumption \ref{As1} with $C_2 = \frac{1}{4}$. In this paper, we use the regularization function $R(f_v) = \frac{1}{2}\parallel f_v \parallel^2$. Then, we can directly use $\mathcal{L}=\mathcal{L}_{lr}$ in Theorem \ref{theorem1} to guarantee the DDP.

\subsection{Convergence of Collaborative Learning}

In the first set of experiments, we test the convergence of Algorithm \ref{algDVP}. The learning performance is measured by the empirical risk (ER). In this experiment, each entire running of the algorithm is based on a fixed value of $\alpha_v(t)$. Figure \ref{conv_1}  and \ref{conv_2} show the test results. As can be seen, larger $\alpha_v(t)$'s lead to faster convergence; for $\alpha_v=0.5$, the converged ER is close to the ER of non-private learning (i.e., Algorithm \ref{Algorithm1}).

\subsection{Security-Privacy Tradeoff}\label{Etradeoff}

%%%%%%%%%%%%%%%%%%%%
\begin{figure}[htpb]
\vspace{-4mm}\includegraphics[scale=0.3]{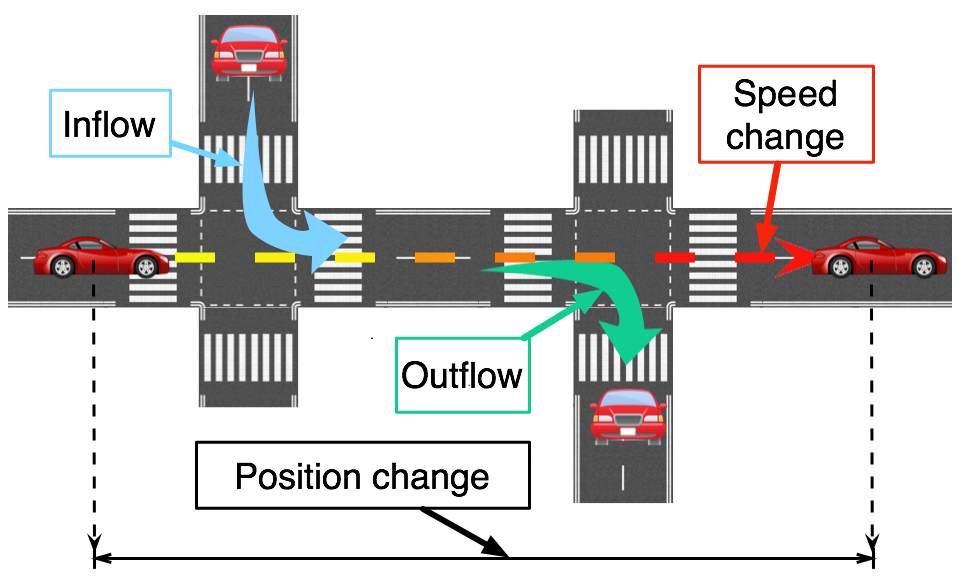}
\centering
\caption{Mobility of VANET: There are three main factors that can cause the changes of the VANET topology, namely, the inflow and the outflow of vehicles, the speed change, and the position change.} \vspace{-0mm}
\label{mobility}
\end{figure}
%%%%%%%%%%%%%%%%%%%%

In this subsection, we explore the tradeoff between the required privacy of training data at each vehicle and the security of the IDS using the classifier trained via the collaborative learning over the VANET. The privacy is quantified by the value of $\alpha_v(t)$. Basically, a larger $\alpha_v(t)$ leads to a larger likelihood ratio 
%
%$\frac{Q(f_v(t)|D_v)}{Q(f_v(t)|D'_v)}$,
$Q(f_v(t)|D_v)/Q(f_v(t)|D'_v)$,
which implies a higher belief of the adversaries about the change of any single entry of the training dataset. Therefore, larger $\alpha_v(t)$'s yields lower privacy. However, the performance of the algorithm increases when the value of $\alpha_v(t)$ grows; higher performance leads to higher level of security. This decreasing monotonicity relationship shows a tradeoff between security and privacy.

We propose an approach to determine an optimal value of $\alpha_v(t)$ that can well manage the security-privacy tradeoff by constructing the utility functions of security and privacy. The design of utility functions at every vehicle $v\in\mathcal{V}$ has to satisfy the conditions stated in the following assumption:
\begin{asu} \label{As6}
The utility of security is monotonically increasing in $\alpha_v(t)$ while the utility of privacy is monotonically decreasing in $\alpha_v(t)$.
\end{asu}

We use the empirical loss (ER), $\overline{J}(t)=\frac{C_1}{n_v}\sum_{i=1}^{n_v}\mathcal{L}_{lr}(y_{iv} f_v(t)^T x_{iv})$, to quantify the security (smaller $\overline{J}(t)$ is, higher the security is). Let $U_{sec}(\cdot):\mathds{R}_+\rightarrow\mathds{R}$ denote utility of security that describes the relationship between $\overline{J}(t)$ and $\alpha_v(t)$. The function $U_{sec}$ is determined by the experimental result, i.e., $(\alpha_v(t), \overline{J}(t))$ using curve fitting. Figure \ref{tradeoffRisk} verifies that $U_{sec}$ (curve fit in green) monotonically decreases with respect to $\alpha_v(t)$ thus the security and $\alpha_v(t)$ has a monotonic decreasing relationship. 

The utility of privacy is designed to meet specific requirements of privacy of each vehicle. Let $U_{pri}(\cdot): \mathds{R}_+\rightarrow \mathds{R}$ represent the utility of privacy. Beside the decreasing monotonicity, $U_{pri}(\cdot)$ is also required to be convex and doubly differentiable such that the optimal value of $\alpha_v(t)$ can be obtained by solving a convex optimization problem. In this experiment, we give an example of utility of privacy defined as 
$$
U_{pri}(\alpha_v(t)) = C_{v1}\cdot\ln\frac{C_{v2}}{C_{v3}\alpha_v(t)+C_{v4}\alpha^2_v(t)},$$
for $\alpha_v(t)\leq 1$, where $C_{vi}\in\mathds{R}_{++}$ for $i=1,\:2,\:3,\:4$. Then, the optimal value of $\alpha_v(t)$ is determined by solving the following problem at iteration $t$:
\begin{equation}\label{equiUtility}
\begin{aligned}
\min_{\alpha_v(t)}& \mathcal{Z}(t)=U_{sec}(\alpha_v(t))-U_{pri}(\alpha_v(t))\\
& s.t. \:\:\: 0<\alpha_v(t)\leq 1, \:\: 0 \leq U_{sec}(\alpha_v(t)) \leq U_1,
\end{aligned}
\end{equation} 
where $U_1$ is the threshold value for $U_{sec}$ beyond which is considered as insecure.

In the experiments, we use a few fixed values of $\rho$ and calculate the empirical loss $\overline{J}(t)=\frac{C_1}{n_v}\sum_{i=1}^{n_v}\mathcal{L}_{lr}(y_{iv} f_v(t)^T x_{iv})$ of the classifier. The value of $\rho$ that gives the minimum $\overline{J}$ for a fixed value of $\alpha_v(t)$ (We use $0.2$ in this experiment). The simulation of Algorithm \ref{Algorithm1}, the non-private algorithm, is used as  the control. In our experiments, we set $\rho = 10^{-9}$ and $10^{-2.5}$ for Algorithm \ref{Algorithm1} and \ref{algDVP}, respectively, and set $C_1 = 1250$ and $650$ for Algorithm \ref{Algorithm1} and \ref{algDVP}, respectively. 

Figure \ref{tradeoffRisk} shows the tradeoff between security and privacy of the DVP at the fixed number of iterations. We model the following function by curve fitting,
$ L_{acc}(\alpha_v(t)) = C_5\cdot e^{-C_6\alpha_v(t)}+ C_7, $
where $C_j\in\mathds{R}_+$, for $j = 5$, $6$, $7$. In our experiment, we determine $C_5=0.055$, $C_6=40$, $C_7=\min_{t} \{\overline{J}(t)\}$
Figure \ref{MisE} shows the receiver operating characteristic (ROC) curves of the outputs of Algorithm \ref{Algorithm1}, and  Algorithm \ref{algDVP} with different values of $\alpha_v(t)$. We can see that when $\alpha_v(t)$ increases, the ROC of the output of the DVP is close to that of the non-private Algorithm \ref{Algorithm1}. This also shows the tradeoff between security and privacy in terms of the ROC.
This feature makes it possible to find an optimal value of $\alpha_v(t)$ such that Algorithm \ref{algDVP} performs similar to Algorithm \ref{Algorithm1}.

%%%%%% Figure 6

% This Figure 6
\begin{figure*}%
\centering
\begin{subfigure}{.5\columnwidth}
\includegraphics[width=\columnwidth]{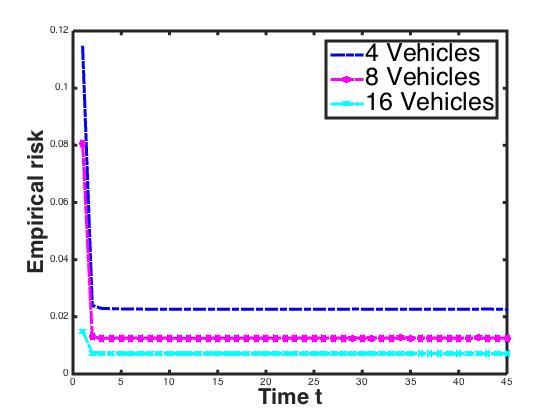}%
\caption{DVP: $\alpha_v = 0.5$.}%
\label{Vsize_c2}%
\end{subfigure}\hfill%
\begin{subfigure}{.5\columnwidth}
\includegraphics[width=\columnwidth]{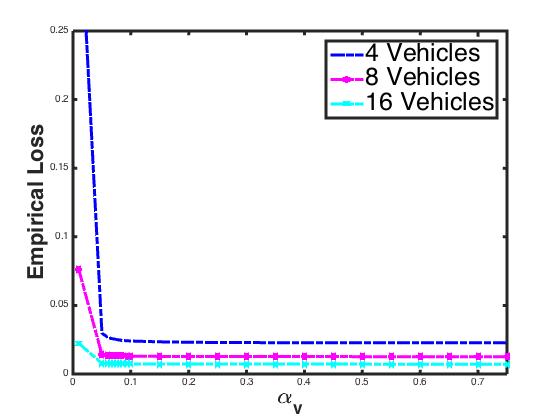}%
\caption{Tradeoff: empirical risk.}%
\label{Vsize_tradeoff}%
\end{subfigure}\hfill%
\begin{subfigure}{.5\columnwidth}
\includegraphics[width=\columnwidth]{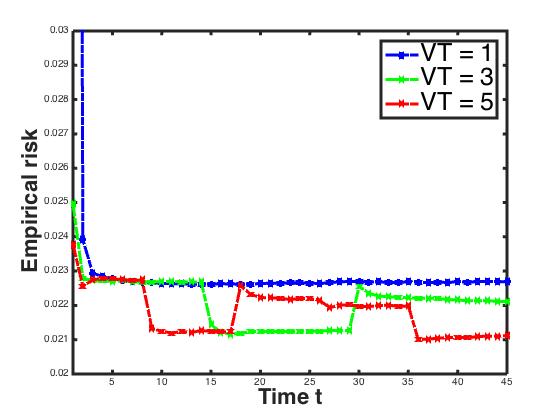}%
\caption{Convergence with $\alpha_v=0.5$.}%
\label{NMV_C}%
\end{subfigure}\hfill%
\begin{subfigure}{.5\columnwidth}
\includegraphics[width=\columnwidth]{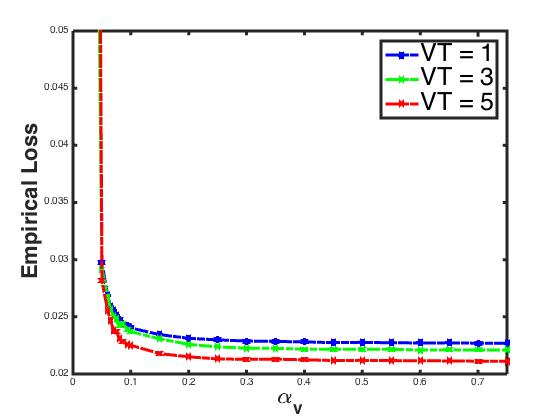}%
\caption{Tradeoff: empirical risk}%
\label{NMV_T}%
\end{subfigure}\hfill%
\caption{Figure \ref{Vsize_c2}-\ref{Vsize_tradeoff}: convergence  with different (fixed) values of $\alpha_v(t)$ with fixed $45$ ADMM iterations; DVP with $\rho = 10^{-2.5}$, $C_1=650$. Figure \ref{Vsize_tradeoff}: security-privacy tradeoff measured by empirical risk for $P=4$, $P=8$, and $P=16$ vehicles. Figure \ref{NMV_C}-\ref{NMV_T}: convergence and security-private tradeoff, respectively, with fixed $P=4$, fixed the total number of iterations as $45$, and different values of $VT$ per collaborative learning; DVP with $\rho = 10^{-2.5}$, $C_1=650$.
}
\label{figTradeoff}\vspace{-2mm}
\end{figure*}

%%%%%%%%%%%%%

%%%%%%%%%%%

% This Figure 7
\begin{figure*}%
\centering
\begin{subfigure}{.5\columnwidth}
\includegraphics[width=\columnwidth]{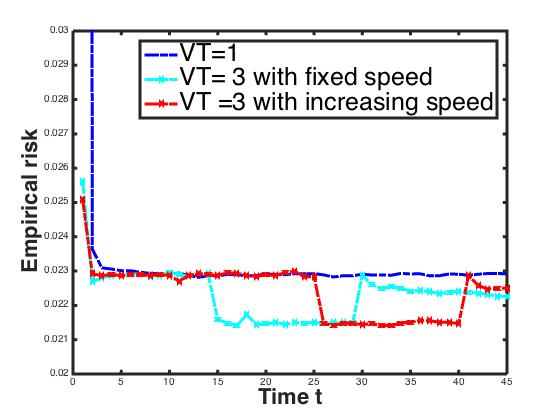}%
\caption{Convergence with $\alpha_v=0.5$}%
\label{speed_C}%
\end{subfigure}\hfill%
\begin{subfigure}{.5\columnwidth}
\includegraphics[width=\columnwidth]{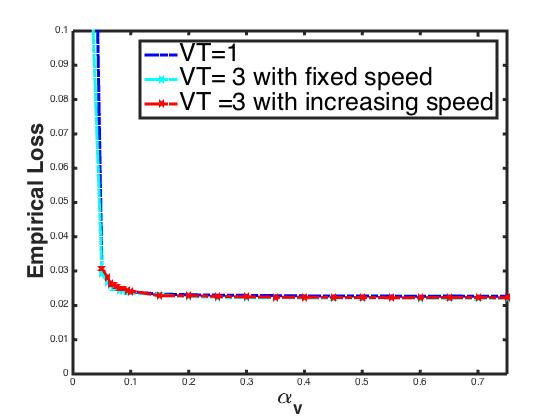}%
\caption{Tradeoff: empirical risk}%
\label{speed_T}%
\end{subfigure}\hfill%
\begin{subfigure}{.5\columnwidth}
\includegraphics[width=\columnwidth]{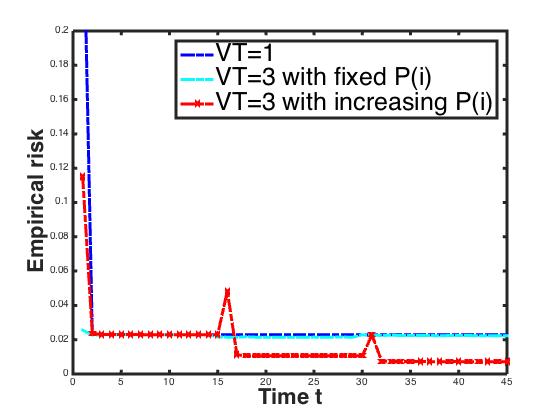}%
\caption{Convergence with $\alpha_v = 0.5$.}%
\label{density_C}%
\end{subfigure}\hfill% 
\begin{subfigure}{.5\columnwidth}
\includegraphics[width=\columnwidth]{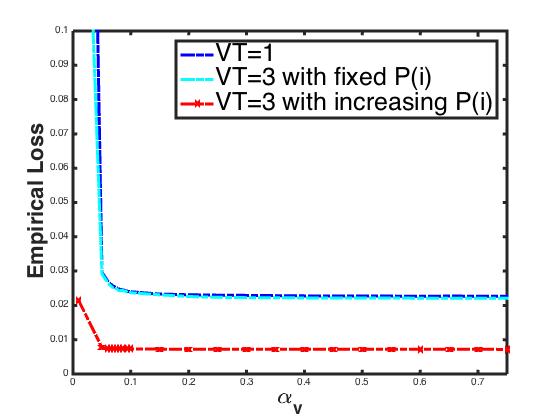}
\caption{Tradeoff: empirical risk}%
\label{density_T}%
\end{subfigure}\hfill%
\caption{Figure \ref{speed_C}-\ref{speed_T}: convergence and security-private tradeoff, respectively, with fixed $P=4$, fixed total number of iterations as $45$, fixed $VT=3$; the result of $VT=1$ is used for comparison; in the case of fixed speed, $k(1)=k(2)=k(3)=15$; in the case of increasing speed, $k(1) = 25$, $k(2) = 15$, $k(3) = 5$; DVP with $\rho = 10^{-2.5}$, $C_1=650$.
Figure \ref{density_C}-\ref{density_T}: convergence and security-private tradeoff, respectively, with a fixed total number of iterations as $45$, fixed $VT=3$, fixed $k(1)=k(2)=k(3)=15$; the result of $VT=1$ is used for comparison; in the case of fixed $P(i)=4$ for $i=1,2,3$, inflow and outflow neutralize each other; in the case of increasing $P(i)$, $P(1)=4$, $P(2)=10$, $P(3)=16$;
DVP with $\rho = 10^{-2.5}$, $C_1=650$.}
\label{figTradeoff}\vspace{-2mm}
\end{figure*}

\subsubsection{Impact of Changing Number of Vehicles in a VANET}
We also investigate how the accuracy changes when we increase the number of vehicles in a VANET. In this experiment, we fix the number of ADMM iterations to $45$ and examine three VANETs with $4$, $8$, and $16$ vehicles, respectively, using NSL-KDD dataset.

Figure \ref{Vsize_c2} shows the convergence results for different $P$. As can be seen, a larger VANET size converges to a smaller value of empirical risk (ER). Figure \ref{Vsize_tradeoff} shows the tradeoff between security and privacy in terms of the ER. We can see that larger $P$ performs better in managing the tradeoff between security and privacy.

\subsection{Topology-Varying VANET}

Due to the high mobility of the vehicles, the topology of a VANET changes frequently. In this paper, we also examine the impact of topology-varying VANET in the security and privacy. 
As shown in Figure \ref{mobility}, the changes of topology can be caused by the inflow and the outflow of the vehicles, the changes of speed, the changes of the positions, or the combinations of these. We focus on the activities of a specific vehicle $v\in\mathcal{V}$. Let $VT$ represent the number of topology changes during the one collaborative learning. Let $V(i)$ be the $i$-th topology for $i=1,...,VT$, and let $k(i)$ be the corresponding number of iterations spent at $V(i)$. Let $P(i)$ be the number of vehicles at the topology $V(i)$. 
$P(i)$ can be changed by the inflow and the outflow. $VT$ can be changed by the changes of the speed and the position. The increase and the decrease of speed also change the values of $k(i)$. 
The collaborative learning scenario at vehicle $v$ over the topology-varying VANET can be described as follows. At $V(i)$, there are $P(i)$ vehicles and  $v$ collaborates with its neighboring vehicles; after $k(i)$ iterations, the VANET topology changes to $V(i+1)$ with $P(i+1)$ vehicles; after $k(i+1)$ iterations, the VANET topology changes to $V(i+2)$, and so on.
In practice, the changes of VANET topology are very fast and complex. We simulate some varying-topology scenarios to study the impact of varying topologies at the outputs of our algorithms.
In the following two sets of experiments, we consider two cases of topology-varying VANET. In Case 1, we fix the value of $P(i)=P$ for $i=1,...,VT$, where $P$ is a constant, while in Case 2, we test the results when $P(i)$ changes for different $i$.

We conduct two experiments in Case 1. In the first experiment, we consider $VT=5$, and fix $P(i) = 4$ and $k(i) = \frac{45}{k}$, which is an integer, for all $i = 1,...., VT$, and we vary the value of $VT$. Figure \ref{NMV_C}-\ref{NMV_T} show the results of convergence and the security-privacy tradeoff for $VT=1$, $VT=3$, and $VT=5$. From Figure \ref{NMV_C}, we can see that the values of ER experience jump when the VANET topology changes. Also, the larger value of $VT$ (more changes of VANET topology) gives smaller values of ER as the number of iterations increases, which is also reflected in Figure \ref{NMV_T}. The tradeoff results shown in Figure \ref{NMV_T} indicate that a larger value of $VT$, i.e., more frequent change of VANET topology, has a better performance in managing the security-privacy tradeoff.

In the second experiment in Case 1, we test the impact on the results when the speed of vehicle changes, which is quantified by the change of $k(i)$. We fix $P(i)=4$, for all $i = 1,...., VT$, and $VT=3$. Figure \ref{speed_C} shows that there is no obvious difference in the value of ER between the fixed speed and the increasing speed. Also, Figure \ref{speed_T} does not show an obvious difference between fixed speed and varying speed regarding managing the security-privacy tradeoff.

In the experiment of Case 2, we fix $VT=3$, and $k(i)=15$, for $i=1,2,3$. We test the results when the number of vehicles $P(i)$ increases as $i$ increases. Figure \ref{density_C} shows the results of convergence. From the convergence results, we can see that a larger density of vehicles per topology has smaller values of ER; also the increasing density of vehicles per topology has smaller values of ER than the fixed-density topology-varying VANET and the fixed-topology VANET. Figure \ref{density_T} shows the behavior of the tradeoff between security and privacy. The results show that increasing $P(i)$ (more inflow of vehicles) outperforms the topology-varying with fixed $P(i)$ and the fixed-topology case.

%%%%%%%%

\section{Conclusion}\label{Se8}

In this paper, we have described an architecture for a collaborative intrusion detection system using privacy-preserving distributed machine learning. The privacy-preserving scheme for the distributed collaborative-based learning is essential for achieving a private collaboration; otherwise, the distributed machine learning itself creates privacy leakage of the training data. We have proposed a privacy-preserving machine-learning based collaborative intrusion detection system (PML-CIDS). The alternating direction method of multipliers (ADMM) approach is used to decentralize the empirical risk minimization (ERM) problem that models the collaborative learning into the distributed ERM well-suited to the nature of the VANET system.

We have proposed the dynamic differential privacy and presented dual variable perturbation (DVP) to protect the privacy of the training data by perturbing the dual variable $\lambda_v(t)$.
We have also analyzed the theoretical performance of the DVP, which is measured by the minimum training data size required to train a classifier with low error.
The tradeoff between security and privacy has been investigated through numerical experiments.
The data used in the experiments is the NSL-KDD dataset. We also have proposed a design principle to select an optimal value of the privacy parameter $\alpha_v(t)$ by solving an optimization problem such that both the security and privacy are optimized. The experiments have also studied the impact of the different VANET size, and the changing VANET topology during the collaborative learning.
As future work, we intend to investigate the collaborative IDS with both supervised and unsupervised machine learning and extend the dynamic differential privacy to different machine learning techniques. We also intend to study the methods of fast incremental learning that can be used in the frequent updates of the machine-learning-based IDSs.

\appendix
\section{Appendix}
\subsection{ Architecture of Vehicular Ad Hoc Network}\label{Apdx_1}

In this appendix, we describe these three main components of the VANET architecture \cite{al2014comprehensive, baldessari2007car}.

\subsubsection{On Board Unit (OBU)} An OBU is a communication device equipped on vehicles for vehicle-to-vehicle and vehicle-to-infrastructure communications through the dedicated short-range communication (DSRC). DSRC is based on IEEE 802.11a technology amended for the low overhead operation to 802.11p \cite{pathan2016security}. The AUs in vehicles use OBUs to communicate with other AUs in other vehicles, and the information exchange is done by OBUs over the ad hoc network. The functions and procedures of an OBU contain wireless radio access, message transfer, geographical ad hoc routing, network congestion control, data security, IP mobility support and others. The basic communication system of an OBU is composed of a minimum set of safety applications, the communication protocol stack with communication transport and network layer, the radio protocols with the IEEE 802.11p devices, and an interface to the local sensors mounted on vehicles. OBUs can also include additional network devices for non-safety applications using other radio technologies like IEEE 802.11a/b/g/n \citep{baldessari2007car}. 

\subsubsection{Application Unit (AU)}
An AU is a device mounted on the vehicle and operates the application installed that uses the communication capabilities of the OBU. Applications can be safety applications such as hazard warning, navigation with communication capabilities, or the Internet-based application such as the personal digital assistant (PDA) \citep{baldessari2007car}.

\subsubsection{Roadside Unit (RSU)}
An RSU is a device located at the fixed positions on the roadside, along highways, or at dedicated locations like parking place or gas stations. Each RSU is equipped with at least one network device for short-range wireless communications based on IEEE 802.11p radio technology. It can also be equipped with other network devices, thereby enabling communications with the infrastructure network. The main functions of an RSU include \citep{baldessari2007car}:
\begin{itemize}
\item Redistributing information to an OBU, thereby extending the communication range of the ad hoc network.
\item Running safety applications such as virtual traffic sign, vehicle-to-infrastructure warnings like accident warning.
\item Enabling the OBUs to connect to the cloud and other infrastructures.
\end{itemize}\vspace{-4mm}

\subsection{PML-CIDS Components}\label{Apdx_2}

In this appendix, we will elaborate on each component of the PML-CIDS architecture.
\subsubsection{Pre-processing Engine}

The pre-processing engine collects and pre-processes real-time audit data flows from OBU and various applications in the AU. These data flows may include user and system activities in the AU and the access request behaviors from outside of the vehicle; it also can be the communication activities between different OBUs or between OBU and RSUs. The pre-processing includes the transformation of symbolic attributes to numerical values, features selection, and data normalization to reduce the possible large variation between values.

\subsubsection{Local Detection Engine}

The local detection engine analyzes the audit data flows processed by the pre-processing engine by a classifier trained via the P-CML engine. If a specific activity is classified as an intrusion, then the local detection engine triggers the alarm.
The user of the vehicle determines how often the classifier needs to be re-trained. If an update of the classifier is required, the local detection engine will initiate the P-CML engine.

\subsubsection{P-CML Engine}

When the P-CML engine is initiated, the ADMM-based private distributed machine learning operates over a temporarily established VANET by collaborating with other vehicles and RSUs. Each vehicle stores its local labeled training dataset. The training dataset can be the historical intrusion activities that have been detected in each vehicle, the data collected by putting sensors on the VANET (for example, getting TCP or Netflow \cite{buczak2015survey}), or the data provided by the trustworthy parties like the Department of Transportation. The participation of RSUs in P-CML enables vehicles to connect to vehicles in distance and other infrastructures such as the cloud.

The P-CML engine is composed of three components, namely, the collaborative communication (CC) engine, the distributed local learning (DLL), and the privacy-preserving (PP) mechanism. The DLL engine updates its local ADMM variables including the dual and the primal variables using its local training dataset and the primal variables transmitted from neighboring vehicles at each ADMM iteration.
The PP engine at each vehicle provides dynamic differential privacy to the VANET involved in the collaboration.
At each ADMM iteration, the CC engine uses the OBU to exchange the intermediately updated parameters with the CC engines of the neighboring vehicles through the low-latency communication \cite{hartenstein2008tutorial} based on the DSRC.

%
%
%

%%%%%%%%%%%%%%%%%%%%%%%%%%%%%%%%%%%
%%%%%%%% CONVERGENCE %%%%%%%%%%%%%%
%%%%%%%%%%%%%%%%%%%%%%%%%%%%%%%%%%%
%Another important performance measure is the convergence of the algorithm. We also investigate the convergence of Algorithm \ref{algDVP} based on the assumption that all the conditions shown in Theorem \ref{theorem3} are satisfied.

%%%%%%%%%%%%%%%%%%%%%%%%%%

\renewcommand\refname{Reference}
\bibliographystyle{ieeetr} 
 \bibliography{VANET_TISPN} 

% \bibliographystyle{unsrt}  
% %\bibliography{references}  %%% Remove comment to use the external .bib file (using bibtex).
% %%% and comment out the ``thebibliography'' section.

% %%% Comment out this section when you \bibliography{references} is enabled.
% \begin{thebibliography}{1}

% \bibitem{kour2014real}
% George Kour and Raid Saabne.
% \newblock Real-time segmentation of on-line handwritten arabic script.
% \newblock In {\em Frontiers in Handwriting Recognition (ICFHR), 2014 14th
%   International Conference on}, pages 417--422. IEEE, 2014.

% \bibitem{kour2014fast}
% George Kour and Raid Saabne.
% \newblock Fast classification of handwritten on-line arabic characters.
% \newblock In {\em Soft Computing and Pattern Recognition (SoCPaR), 2014 6th
%   International Conference of}, pages 312--318. IEEE, 2014.

% \bibitem{hadash2018estimate}
% Guy Hadash, Einat Kermany, Boaz Carmeli, Ofer Lavi, George Kour, and Alon
%   Jacovi.
% \newblock Estimate and replace: A novel approach to integrating deep neural
%   networks with existing applications.
% \newblock {\em arXiv preprint arXiv:1804.09028}, 2018.

% \end{thebibliography}

\end{document}